\documentclass[runningheads,a4paper]{llncs}

\usepackage{amssymb}
\usepackage{graphicx}
\usepackage{url}
\usepackage{amsmath}
\usepackage{tikz}
\usetikzlibrary{calc}
\usetikzlibrary{arrows.meta,positioning}

\begin{document}

\mainmatter 


\title{Additional properties of parity based bit-counting complexity classes and hierarchies}

\titlerunning{Parity based bit-counting complexity classes and hierarchies}

\author{ Tayfun Pay}

\authorrunning{Tayfun Pay}

\urldef{\mailsa}\path|tpay@gradcenter.cuny.edu|

\institute{{ \mailsa}}



%
%

\maketitle

\begin{abstract}
We study some properties of the parity based bit-counting complexity classes ${\bf B_{|0| \oplus}P}$ and ${\bf B_{|1| \oplus}P}$. We first prove that both of these complexity classes are closed under complement and ${\bf B_{|1|\oplus}P}\subseteq {\bf B_{|0|\oplus}P}$. We then prove that ${\bf US}\subseteq {\bf P}^{{\bf B_{|1|\oplus}P}}$ and ${\bf US}\subseteq {\bf P}^{{\bf B_{|0|\oplus}P}}$. We then study the class defining characteristic functions of the parity based bit-counting complexity classes, where the one associated with ${\bf B_{|1| \oplus}P}$ produces the Prouhet-Thue-Morse sequence. We then prove that a contiguous block of four values from either sequence determines the parity of its starting index and use this fact to show that ${\bf \oplus P}\subseteq {\bf P}^{{\bf B_{|0|\oplus}P}}$ and ${\bf \oplus P}\subseteq {\bf P}^{{\bf B_{|1|\oplus}P}}$. We then use the parity based bit-counting complexity classes to define various hierarchies and show that they all contain ${\bf PH}$ and are contained in ${\bf CH}$. 
\end{abstract}

\section{Introduction}

Classical counting complexity classes impose conditions directly on the number of accepting paths of a non-deterministic polynomial time Turing machine. Bit-counting complexity classes instead examine the binary representation of that accepting path count and impose conditions on properties of its bits. As a result, the acceptance condition depends not only on the numerical value of a ${\bf \# P}$ \cite{V79} function, but also on the structure of its binary representation. This provides a different way of extracting information from ${\bf \#P}$ functions and leads to complexity classes whose behavior can differ substantially from more familiar counting classes.

Among the bit-counting complexity classes introduced in \cite{P26}, the parity based bit-counting complexity classes ${\bf B_{|0| \oplus}P}$ and ${\bf B_{|1| \oplus}P}$ are particularly natural to compare. Although the number of 0’s bits and the number of 1’s bits are closely related within a binary representation of fixed length, the standard binary representation has varying length and does not contain leading zeros. It is therefore not immediate that parity conditions on the number of 0’s bits and on the number of 1’s bits yield equivalent complexity classes. This in particular motivates the study of complement closures and direct containments of parity based bit-counting complexity classes. In essence, we show that the parity based bit-counting complexity classes ${\bf B_{|0| \oplus}P}$ and ${\bf B_{|1| \oplus}P}$ are both closed under complement. We also improve upon the containment proven in \cite{P26}, namely ${\bf B_{|1|\oplus}P}\subseteq {\bf P}^{{\bf B_{|0|\oplus}P}}$, and prove that ${\bf B_{|1|\oplus}P}\subseteq {\bf B_{|0|\oplus}P}$. 

We then examine the relationship between parity based bit-counting complexity classes and the classical complexity class  ${\bf US}$ \cite{BG82} and prove that ${\bf US}\subseteq{\bf P}^{{\bf B{|1|\oplus}P}}$ and ${\bf US}\subseteq{\bf P}^{{\bf B{|0|\oplus}P}}$. 

We next explore the relationship between the parity based bit-counting complexity classes and the classical parity based complexity class ${\bf \oplus P}$ \cite{PZ83}. Note that the acceptance criteria of ${\bf B_{|0| \oplus}P}$ and ${\bf B_{|1| \oplus}P}$ define natural infinite binary sequences over the nonnegative integers, which correspond to their class defining characteristic functions. In the case of ${\bf B_{|1| \oplus}P}$, the resulting sequence is the well known Prouhet-Thue-Morse sequence \cite{M21}. In a way, this creates a connection between bit-counting complexity and combinatorial properties of binary sequences. We then prove that for both of these sequences, a finite contiguous block of size four is enough to determine whether its starting index is even or odd. We then use these facts to construct deterministic polynomial time oracle machines that determine whether the relevant ${\bf \#P}$ function is even or odd by querying the corresponding parity based bit-counting sequence at four consecutive values. This proves that ${\bf \oplus P}\subseteq {\bf P}^{{\bf B_{|0|\oplus}P}}$ and ${\bf \oplus P}\subseteq {\bf P}^{{\bf B_{|1|\oplus}P}}$.

We then investigate what happens when the parity based bit-counting complexity classes themselves are iterated as oracle operators. Essentially, we introduce four parity based bit-counting hierarchies inspired by the polynomial hierarchy, ${\bf PH}$ \cite{S76}. The ${\bf \Sigma B_{|0|\oplus}H}$ and ${\bf \Sigma B_{|1|\oplus}H}$ hierarchies that we introduce use the parity based bit-counting complexity classes ${\bf B_{|0| \oplus}P}$ and ${\bf B_{|1| \oplus}P}$ respectively, and mirror the $\Sigma$ side of the ${\bf PH}$. The ${\bf \Delta B_{|0|\oplus}H}$ and ${\bf \Delta B_{|1|\oplus}H}$ hierarchies that we introduce once again use the parity based bit-counting complexity classes ${\bf B_{|0| \oplus}P}$ and ${\bf B_{|1| \oplus}P}$ respectively, and mirror the $\Delta$ part of the ${\bf PH}$. We then show that ${\bf \Sigma B_{|0|\oplus}H}$ equals ${\bf \Sigma B_{|1|\oplus}H}$ and ${\bf \Delta B_{|0|\oplus}H}$ equals ${\bf \Delta B_{|1|\oplus}H}$, after taking their unions over all levels.  After that, we prove that all of the parity based bit-counting hierarchies that we introduce contain ${\bf PH}$. 

We then define the alternating variants of these hierarchies as ${\bf Alt\Sigma B_{|0|\oplus}H}$, ${\bf Alt\Sigma B_{|1|\oplus}H}$, ${\bf Alt\Delta B_{|0|\oplus}H}$ and ${\bf Alt\Delta B_{|1|\oplus}H}$. These hierarchies have the opposite parity based bit-counting complexity class at their even valued levels. We show that ${\bf Alt\Sigma B_{|0|\oplus}H}$ equals ${\bf Alt\Sigma B_{|1|\oplus}H}$ and that ${\bf Alt\Delta B_{|0|\oplus}H}$ equals ${\bf Alt\Delta B_{|1|\oplus}H}$, after taking unions over all of their levels. We also show that the alternating parity based bit-counting hierarchies equal their non-alternating counterparts. Finally, we show that all of the parity based bit-counting hierarchies are contained in the counting hierarchy ${\bf CH}$ \cite{W86}.

These results illustrate the importance of the parity based bit-counting complexity classes. In particular, ${\bf B_{|0|\oplus}P}$ and ${\bf B_{|1|\oplus}P}$ simultaneously contain both ${\bf NP}$ and ${\bf CoNP}$ and, when provided as oracles to deterministic polynomial time Turing machines, allow every language in ${\bf US}$ and ${\bf \oplus P}$ to be decided. We also establish the existence of well defined hierarchies generated by the parity based bit-counting complexity classes that contain the polynomial hierarchy ${\bf PH}$ and are themselves contained in the counting hierarchy ${\bf CH}$, and therefore occupy a natural position between these two classical hierarchies with respect to containment.

\section{Definitions and containments}

\subsection{Some classical complexity classes and hierarchies}

\begin{definition}\normalfont
A language $L$ is in complexity class {\bf NP}, if there exists a polynomial $p$ and a polynomial time predicate $R$ such that, for each $x$,

$x \in L \Leftrightarrow ||\{y| \ |y| = p(|x|) \wedge R(x, y)\}|| >0 $
\end{definition}

\begin{definition}\normalfont
A language $L$ is in complexity class {\bf CoNP}, if there exists a polynomial $p$ and a polynomial time predicate $R$ such that, for each $x$,

$x \in L \Leftrightarrow ||\{y| \ |y| = p(|x|) \wedge R(x, y)\}|| =0 $
\end{definition}

\begin{definition}\normalfont
A language $L$ is in complexity class {\bf US}, as defined in \cite{BG82}, if there exists a polynomial $p$ and a polynomial time predicate $R$ such that, for each $x$,

$x \in L \Leftrightarrow ||\{y| \ |y| = p(|x|) \wedge R(x, y)\}|| = 1 $
\end{definition}

\begin{definition}\normalfont
The polynomial hierarchy, denoted ${\bf PH}$ and as defined in \cite{S76}, is the union of the levels ${\bf \Sigma_k^P}$, ${\bf \Pi_k^P}$, and ${\bf \Delta_k^P}$. The zeroth levels are ${\bf \Sigma_0^P}={\bf \Pi_0^P}={\bf \Delta_0^P}={\bf P}$. For every $k\geq 0$, define ${\bf \Sigma_{k+1}^P}={\bf NP}^{{\bf \Sigma_k^P}}$, ${\bf \Pi_{k+1}^P}={\bf coNP}^{{\bf \Sigma_k^P}}$, and ${\bf \Delta_{k+1}^P}={\bf P}^{{\bf \Sigma_k^P}}$. The full polynomial hierarchy is ${\bf PH}=\bigcup_{k\geq 0}{\bf \Sigma_k^P}=\bigcup_{k\geq 0}{\bf \Pi_k^P}=\bigcup_{k\geq 0}{\bf \Delta_k^P}$.
\end{definition}

The first few levels of the ${\bf PH}$ are depicted in the picture below.

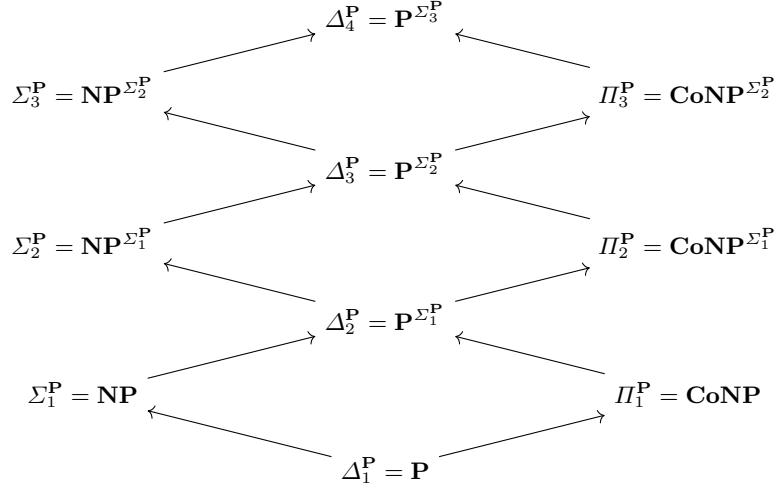
\begin{tikzpicture}[
    every node/.style={font=\small},
    classnode/.style={align=center}
]

\node[classnode] (d1) at (0,0) {$\Delta_1^{\bf P}={\bf P}$};

\node[classnode] (s1) at (-4,1) {$\Sigma_1^{\bf P}={\bf NP}$};
\node[classnode] (p1) at (4,1) {$\Pi_1^{\bf P}={\bf CoNP}$};

\node[classnode] (d2) at (0,2) {$\Delta_2^{\bf P}={\bf P}^{\Sigma_1^{\bf P}}$};

\node[classnode] (s2) at (-4,3) {$\Sigma_2^{\bf P}={\bf NP}^{\Sigma_1^{\bf P}}$};
\node[classnode] (p2) at (4,3) {$\Pi_2^{\bf P}={\bf CoNP}^{\Sigma_1^{\bf P}}$};

\node[classnode] (d3) at (0,4) {$\Delta_3^{\bf P}={\bf P}^{\Sigma_2^{\bf P}}$};

\node[classnode] (s3) at (-4,5) {$\Sigma_3^{\bf P}={\bf NP}^{\Sigma_2^{\bf P}}$};
\node[classnode] (p3) at (4,5) {$\Pi_3^{\bf P}={\bf CoNP}^{\Sigma_2^{\bf P}}$};

\node[classnode] (d4) at (0,6) {$\Delta_4^{\bf P}={\bf P}^{\Sigma_3^{\bf P}}$};

\draw[->] (d1) -- (s1);
\draw[->] (d1) -- (p1);

\draw[->] (s1) -- (d2);
\draw[->] (p1) -- (d2);

\draw[->] (d2) -- (s2);
\draw[->] (d2) -- (p2);

\draw[->] (s2) -- (d3);
\draw[->] (p2) -- (d3);

\draw[->] (d3) -- (s3);
\draw[->] (d3) -- (p3);

\draw[->] (s3) -- (d4);
\draw[->] (p3) -- (d4);

\end{tikzpicture}

\begin{definition}\normalfont
A language $L$ is in complexity class {\bf $\oplus$P}, as defined in \cite{PZ83}, if there exists a polynomial $p$ and a polynomial time predicate $R$ such that, for each $x$,

$x \in L \Leftrightarrow ||\{y| \ |y| = p(|x|) \wedge R(x, y)\}|| \not \equiv$ {\rm 0 (Mod 2)}
\end{definition}

\begin{definition}\normalfont
A language $L$ is in complexity class {\rm \bf C$_{=}$P}, as defined in \cite{S75}, if there exists a polynomial $p$ and a polynomial time predicate $R$ such that, for each $x$,
$x \in L \Leftrightarrow ||\{y| \ |y| = p(|x|) \wedge R(x, y)\}|| = 2^{p(|x|)-1} $
\end{definition}

\begin{definition}\normalfont
A language $L$ is in complexity class {\rm \bf ES}, as defined in \cite{BHR00}, if there exists a polynomial $p$ and a polynomial time predicate $R$ such that, for each $x$,

$x \in L \Leftrightarrow ||\{y| \ |y| = p(|x|) \wedge R(x, y)\}|| = 2^{t}$, where t $\in \mathbb{N}_{0} = \{0,1,2, ...\}$
\end{definition}

\begin{definition}\normalfont A language $L$ is in complexity class {\bf MNS}, as defined in \cite{CP18}, if there exists a polynomial $p$ and a polynomial time predicate $R$ such that, for each $x$,

$x \in L \Leftrightarrow ||\{y| \ |y| = p(|x|) \wedge R(x, y)\}|| = 2^{t}-1$, where t $\in \mathbb{N}_{>0} = \{1,2,3, ...\}$
\end{definition}

\begin{definition}\normalfont
A language $L$ is in complexity class {\rm \bf PP}, as defined in \cite{S75}, if there exists a polynomial $p$ and a polynomial time predicate $R$ such that, for each $x$,

$x \in L \Leftrightarrow ||\{y| \ |y| = p(|x|) \wedge R(x, y)\}|| > 2^{p(|x|)-1} $

\end{definition}

\begin{definition}\normalfont

The counting hierarchy, denoted ${\bf CH}$, was introduced in \cite{W86}. We use the oracle characterization of the counting hierarchy defined in \cite{T88}. The counting hierarchy is the union of the levels ${\bf C_kP}$. The zeroth level is ${\bf C_0P}={\bf P}$. For every $k\geq 0$, the next level is ${\bf C_{k+1}P}={\bf PP}^{{\bf C_kP}}$, where ${\bf PP}^{{\bf C_kP}}$ denotes the class of languages decidable by a ${\bf PP}$ machine with oracle access to a language in ${\bf C_kP}$. The full counting hierarchy is ${\bf CH}=\bigcup_{k\geq 0}{\bf C_kP}$. 
\end{definition}

The first few levels of the ${\bf CH}$ are ${\bf C_0P}={\bf P}$, ${\bf C_1P}={\bf PP}$, ${\bf C_2P}={\bf PP}^{\bf PP}$, and ${\bf C_3P}={\bf PP}^{{\bf PP}^{\bf PP}}$.

\begin{definition}\normalfont
The functional complexity class {\rm \bf \#P}, as defined in \cite{V79}, counts the total number of accepting paths of a non-deterministic polynomial time Turing machine.

{\bf \#P} = $\{f | (\exists$ a non-deterministic polynomial time Turing machine $M) (\forall x)$ $[f (x) = \#accept_{M}(x)]\}$.
\end{definition}

It is known that ${\bf CoNP} \subseteq {\bf US} \subseteq {\bf MNS} = {\bf ES } = { \bf C_{=}P} \subseteq {\bf PP}$ and that ${\bf PH}\subseteq{\bf P}^{\bf PP} = {\bf P}^{{\bf\#P}} \subseteq {\bf CH}$, where it was shown in \cite{T89} that ${\bf PH}\subseteq{\bf P}^{\bf PP}$. Numerous other complexity classes can be found in \cite{HO02} and \cite{CP18}.

\subsection{Bit-counting complexity classes}
The function $B_{0}$ counts the number of 0's bits and the function $B_{1}$ counts the number of 1's bits in the binary representation of the numbers in the following definitions. Note that $B_{0}(0)=1$. 

\begin{definition}\normalfont
A language $L$ is in complexity class ${\bf B_{|0|=|1|}P}$, if there exist a polynomial $p$ and a polynomial time predicate $R$ such that, for each $x$, 

$x \in L \Leftrightarrow B_{0}(||\{y| \ |y| = p(|x|) \wedge R(x, y)\}||)$ $=$ $B_{1}(||\{y| \ |y| = p(|x|) \wedge R(x, y)\}||)$
\end{definition}

\begin{definition}\normalfont
A language $L$ is in complexity class ${\bf B_{|0|>|1|}P}$, if there exist a polynomial $p$ and a polynomial time predicate $R$ such that, for each $x$, 

$x \in L \Leftrightarrow B_{0}(||\{y| \ |y| = p(|x|) \wedge R(x, y)\}||)$ $>$ $B_{1}(||\{y| \ |y| = p(|x|) \wedge R(x, y)\}||)$
\end{definition}

\begin{definition}\normalfont
A language $L$ is in complexity class ${\bf B_{|0|<|1|}P}$, if there exist a polynomial $p$ and a polynomial time predicate $R$ such that, for each $x$, 

$x \in L \Leftrightarrow B_{0}(||\{y| \ |y| = p(|x|) \wedge R(x, y)\}||)$ $<$ $B_{1}(||\{y| \ |y| = p(|x|) \wedge R(x, y)\}||)$
\end{definition}

\begin{definition}\normalfont
A language $L$ is in complexity class ${\bf B_{|0| \oplus}P}$, if there exist a polynomial $p$ and a polynomial time predicate $R$ such that, for each $x$,  

$x \in L \Leftrightarrow B_{0}(||\{y| \ |y| = p(|x|) \wedge R(x, y)\}||) \not \equiv  {\rm 0 (Mod 2)}$
\end{definition}

\begin{definition}\normalfont
A language $L$ is in complexity class ${\bf B_{|1| \oplus}P}$, if there exist a polynomial $p$ and a polynomial time predicate $R$ such that, for each $x$,  

$x \in L \Leftrightarrow B_{1}(||\{y| \ |y| = p(|x|) \wedge R(x, y)\}||) \not \equiv  {\rm 0 (Mod 2)}$
\end{definition}

\begin{definition}\normalfont
A language $L$ is in complexity class ${\bf B_{|1|=0}P}$, if there exist a polynomial $p$ and a polynomial time predicate $R$ such that, for each $x$,
 
$x \in L \Leftrightarrow B_{1}(||\{y| \ |y| = p(|x|) \wedge R(x, y)\}||)$ $= 0$
\end{definition}

\begin{definition}\normalfont
A language $L$ is in complexity class ${\bf B_{|1|>0}P}$, if there exist a polynomial $p$ and a polynomial time predicate $R$ such that, for each $x$,
 
$x \in L \Leftrightarrow B_{1}(||\{y| \ |y| = p(|x|) \wedge R(x, y)\}||)$ $> 0$
\end{definition}

\begin{definition}\normalfont
A language $L$ is in complexity class ${\bf B_{|0|=0}P}$, if there exist a polynomial $p$ and a polynomial time predicate $R$ such that, for each $x$, 

$x \in L \Leftrightarrow B_{0}(||\{y| \ |y| = p(|x|) \wedge R(x, y)\}||)$ $= 0$
\end{definition}

\begin{definition}\normalfont
A language $L$ is in complexity class ${\bf B_{|0|>0}P}$, if there exist a polynomial $p$ and a polynomial time predicate $R$ such that, for each $x$, 

$x \in L \Leftrightarrow B_{0}(||\{y| \ |y| = p(|x|) \wedge R(x, y)\}||)$ $> 0$
\end{definition}

\subsection{Some known containments and equalities }

The following equalities follow from their respective definitions.  

-${\bf B_{|1|=0}P} ={\bf CoNP}$.

-${\bf B_{|1|>0}P} ={\bf NP}$.

-${\bf B_{|0|=0}P} = {\bf MNS}$.
 
-${\bf B_{|0|>0}P} = {\bf CoMNS}$.  
\newline \newline
The following containments and equalities were proven in \cite{P26}. 

-${\bf NP}\subseteq {\bf B_{|1| \oplus}P}$

-${\bf NP}\subseteq {\bf B_{|0| \oplus}P}$

-${\bf CoNP}\subseteq {\bf B_{|1| \oplus}P}$

-${\bf CoNP}\subseteq {\bf B_{|0| \oplus}P}$

-${\bf B_{|1| \oplus}P}\subseteq {\bf P}^{\bf PP}$

-${\bf B_{|0| \oplus}P}\subseteq {\bf P}^{\bf PP}$

-${\bf B_{|1|\oplus}P}\subseteq {\bf P}^{{\bf B_{|0|\oplus}P}}$

-${\bf B_{|0|\oplus}P}\subseteq {\bf P}^{{\bf B_{|1|\oplus}P}}$

-${\bf P}^{{\bf B_{|0|\oplus}P}}={\bf P}^{{\bf B_{|1|\oplus}P}}$

-${\bf PP}\subseteq {\bf B_{|0|=|1|}P}$

-${\bf PP}\subseteq {\bf B_{|0|>|1|}P}$

-${\bf PP}\subseteq {\bf B_{|0|<|1|}P}$

-${\bf B_{|0|=|1|}P}\subseteq {\bf P}^{\bf PP}$

-${\bf B_{|0|>|1|}P}\subseteq {\bf P}^{\bf PP}$

-${\bf B_{|0|<|1|}P}\subseteq {\bf P}^{\bf PP}$

-${\bf P}^{\bf PP} = {\bf P}^{{\bf B_{|0|=|1|}P}} = {\bf P}^{{\bf B_{|0|>|1|}P}}={\bf P}^{{\bf B_{|0|<|1|}P}}$

\section{Some properties of parity based bit-counting complexity classes}

\subsection{Closure under complement and containment}
We show that both of the parity based bit-counting complexity classes ${\bf B_{|0|\oplus}P}$ and ${\bf B_{|1|\oplus}P}$ are closed under complement. We also show that  ${\bf B_{|1|\oplus}P}\subseteq {\bf B_{|0|\oplus}P}$, where the previous best result was ${\bf B_{|1|\oplus}P}\subseteq {\bf P}^{\bf B_{|0|\oplus}P}$.

\begin{theorem}
${\bf B_{|0|\oplus}P}$ is closed under complement.
\end{theorem}

\begin{proof}
Let $L\in{\bf B_{|0|\oplus}P}$. Then there exists a ${\bf \#P}$ function $H(x)$ such that $x\in L \iff B_0(H(x))\not\equiv 0\pmod 2$. Define $G(x)=4H(x)+1$. Since $H(x)$ is a ${\bf \#P}$ function, $G(x)$ is also a ${\bf \#P}$ function.

We claim that $B_0(G(x))$ has the opposite parity from $B_0(H(x))$. If $H(x)=0$, then $G(x)=1$, so $B_0(H(x))=B_0(0)=1$ and $B_0(G(x))=B_0(1)=0$. 

If $H(x)>0$, then multiplying by four appends two zeros to the binary representation of $H(x)$, and adding one changes the final two bit suffix from $00$ to $01$. As a result, the binary representation of $G(x)=4H(x)+1$ is obtained from the binary representation of $H(x)$ by appending $01$. Therefore, $B_0(G(x))=B_0(H(x))+1$.

Thus, for every $x$, $B_0(G(x))\not\equiv 0\pmod 2$ if and only if $B_0(H(x))\equiv 0\pmod 2$. Since $x\in L$ exactly when $B_0(H(x))$ is odd, it follows that $x\in\overline L$ exactly when $B_0(G(x))$ is odd. As a result, $\overline L\in{\bf B_{|0|\oplus}P}$. Therefore, ${\bf B_{|0|\oplus}P}$ is closed under complement.$\qed$
\end{proof}

\begin{theorem}
${\bf B_{|1|\oplus}P}$ is closed under complement.
\end{theorem}

\begin{proof}
Let $L\in{\bf B_{|1|\oplus}P}$. Then there exists a ${\bf \#P}$ function $H(x)$ such that $x\in L \iff B_1(H(x))\not\equiv 0\pmod 2$. Define $G(x)=2H(x)+1$. Since $H(x)$ is a ${\bf \#P}$ function, $G(x)$ is also a ${\bf \#P}$ function.

We claim that $B_1(G(x))$ has the opposite parity from $B_1(H(x))$. If $H(x)=0$, then $G(x)=1$, so $B_1(H(x))=0$ and $B_1(G(x))=1$. If $H(x)>0$, then multiplying by two appends a $0$ to the binary representation of $H(x)$, and adding one changes it to a $1$. As a result, the binary representation of $G(x)=2H(x)+1$ is obtained from the binary representation of $H(x)$ by appending $1$. Therefore, $B_1(G(x))=B_1(H(x))+1$.

Thus, for every $x$, $B_1(G(x))\not\equiv 0\pmod 2$ if and only if $B_1(H(x))\equiv 0\pmod 2$. Since $x\in L$ exactly when $B_1(H(x))$ is odd, it follows that $x\in\overline{L}$ exactly when $B_1(G(x))$ is odd. As a result, $\overline{L}\in{\bf B_{|1|\oplus}P}$. Therefore, ${\bf B_{|1|\oplus}P}$ is closed under complement.$\qed$
\end{proof}

\begin{theorem}

${\bf B_{|1|\oplus}P}\subseteq {\bf B_{|0|\oplus}P}$

\end{theorem}

\begin{proof}
We use the following definition of ${\bf B_{|0|\oplus}P}$. A language $L$ is in ${\bf B_{|0|\oplus}P}$ if one can construct a ${\bf \#P}$ function $H(x)$ such that $x\in L \iff B_0(H(x))\not\equiv 0\pmod 2$, where $B_0(n)$ denotes the number of $0$'s bits in the binary representation of $n$.

Let $L\in{\bf B_{|1|\oplus}P}$. Then there exists a ${\bf \#P}$ function $F(x)$ such that $x\in L \iff B_1(F(x))\not\equiv 0\pmod 2$, where $B_1(n)$ denotes the number of $1$'s bits in the binary representation of $n$. Since $F(x)$ is a ${\bf \#P}$ function, there is a polynomial $p(n)$ such that $0\leq F(x)\leq 2^{p(|x|)}$ for every input $x$. We will use $p=p(|x|)$ for readability. We define $q=2p+2$, which is even.  We also have $F(x)<2^q$, since $F(x)\leq 2^p$, and $p<q$. We next define $H(x)=2^q+F(x)$.

We first verify that $H(x)$ is a valid ${\bf \#P}$ function. The length dependent term $2^q$ is a ${\bf \#P}$ function since $q$ is polynomially bounded. Since ${\bf \#P}$ is closed under addition, we have $H(x)=2^q+F(x)\in{\bf \#P}$.

We next prove the main theorem with the following two lemmas.

\begin{lemma}
If $x\in L$, then $B_0(H(x))\not\equiv 0\pmod 2$.
\end{lemma}

\begin{proof}
Assume $x\in L$. Then $B_1(F(x))\not\equiv 0\pmod 2$. Since $F(x)<2^q$, the sum $H(x)=2^q+F(x)$ has a $1$ in bit position $q$, and its lower bit positions $q-1,q-2,\ldots,0$ have value $F(x)$. The bit in position $q$ is $1$, so it contributes no $0$'s bits. Among the lower $q$ positions, exactly $B_1(F(x))$ positions contain a $1$'s bit. Therefore, the remaining $q-B_1(F(x))$ lower positions contain $0$'s bits. Thus, $B_0(H(x))=q-B_1(F(x))$.

Since $q$ is even and $B_1(F(x))$ is odd, then $q-B_1(F(x))$ is odd. Therefore, $B_0(H(x))\not\equiv 0\pmod 2$.$\qed$

\end{proof}

\begin{lemma}
If $x\notin L$, then $B_0(H(x))\equiv 0\pmod 2$.
\end{lemma}

\begin{proof}
Assume $x\notin L$. Then $B_1(F(x))\equiv 0\pmod 2$. Since $F(x)<2^q$, the sum $H(x)=2^q+F(x)$ has a $1$ in bit position $q$, and its lower bit positions $q-1,q-2,\ldots,0$ have value $F(x)$. The bit in position $q$ is $1$, so it contributes no $0$'s bits. Among the lower $q$ positions, exactly $B_1(F(x))$ positions contain $1$'s bits. Therefore, the remaining $q-B_1(F(x))$ lower positions contain $0$'s bits. Thus, $B_0(H(x))=q-B_1(F(x))$.

Since $q$ is even and $B_1(F(x))$ is even, then $q-B_1(F(x))$ is even. Therefore, $B_0(H(x))\equiv 0\pmod 2$.$\qed$

\end{proof}

The two lemmas show that $x\in L \iff B_0(H(x))\not\equiv 0\pmod 2$. Since $H(x)\in{\bf \#P}$, we have that $L\in{\bf B_{|0|\oplus}P}$. Since $L\in{\bf B_{|1|\oplus}P}$ was any $L$, then we can conclude that ${\bf B_{|1|\oplus}P}\subseteq{\bf B_{|0|\oplus}P}$.$\qed$

\end{proof}

\subsection{Containment of ${\bf US}$}

We show that ${\bf US}\subseteq {\bf P}^{{\bf B_{|1|\oplus}P}}$ and ${\bf US}\subseteq {\bf P}^{{\bf B_{|0|\oplus}P}}$. Each proof uses two different functions. The first function is used to determine whether the number of accepting paths is zero. The second function is used to determine whether the number of accepting paths is greater than one.

\begin{theorem}
${\bf US}\subseteq {\bf P}^{{\bf B_{|1|\oplus}P}}$
\end{theorem}

\begin{proof}
Let $L\in{\bf US}$. Then there is a non-deterministic polynomial time Turing machine $M$ such that, for every input $x$, $x\in L$ if and only if $M(x)$ has exactly one accepting path. Let $A(x)=\#\operatorname{acc}_M(x)$. Then $A(x)$ is a ${\bf \#P}$ function, and $x\in L$ if and only if $A(x)=1$. The machine $M$ uses exactly $p(n)$ non-deterministic bits on inputs of length $n$. We will use $p=p(|x|)$ for readability. Thus, $0\leq A(x)\leq 2^p$. Set $r=2p+1$. Then $r$ is odd. We also have $A(x)<2^r$ and $A(x)(A(x)-1)<2^r$.

We now define $H_0(x)=(2^r-1)A(x)$. This is a ${\bf \#P}$ function since a non-deterministic machine can guess an accepting path of $M(x)$ and also guess a nonzero $r$-bit string. If $A(x)=0$, then $H_0(x)=0$, so $B_1(H_0(x))=0$ is even. If $A(x)>0$, then $1\leq A(x)<2^r$, and the binary representation of $(2^r-1)A(x)$ has exactly $r$ many $1$'s bits. Note that $(2^r-1)A(x)=(A(x)-1)2^r+(2^r-A(x))$. Moreover, $2^r-A(x)=((2^r-1)-(A(x)-1))$. So the $r$-bit representations of $A(x)-1$ and $2^r-A(x)$ are bitwise complements. Since $r$ is odd, $B_1(H_0(x))\not\equiv 0\pmod 2$. Therefore, $B_1(H_0(x))\not\equiv 0\pmod 2$ if and only if $A(x)>0$.

We next define $P(x)=A(x)(A(x)-1)$. This is a ${\bf \#P}$ function because it counts ordered pairs of distinct accepting paths of $M(x)$. We have $P(x)>0$ if and only if $A(x)\geq 2$. Define $H_1(x)=(2^r-1)P(x)$. This is also a ${\bf \#P}$ function since a non-deterministic machine can guess an ordered pair of distinct accepting paths of $M(x)$ and also guess a nonzero $r$-bit string. If $P(x)=0$, then $H_1(x)=0$, so $B_1(H_1(x))=0$ is even. If $P(x)>0$, then $1\leq P(x)<2^r$ and the binary representation of $(2^r-1)P(x)$ has exactly $r$ many $1$'s bits. Since $r$ is odd, $B_1(H_1(x))\not\equiv 0\pmod 2$. Therefore, $B_1(H_1(x))\not\equiv 0\pmod 2$ if and only if $A(x)\geq 2$.

By tagging the query with the value of $j\in\{0,1\}$, the two functions $H_0$ and $H_1$ can be combined into a single oracle language $O\in{\bf B_{|1|\oplus}P}$. The oracle says yes on the tagged query $(x,0)$ exactly when $B_1(H_0(x))\not\equiv 0\pmod 2$, and it says yes on the tagged query $(x,1)$ exactly when $B_1(H_1(x))\not\equiv 0\pmod 2$.

The deterministic polynomial time oracle machine queries $O$ on $(x,0)$ and $(x,1)$. Let the two oracle answer bits be $a_0$ and $a_1$. By the construction above, $a_0=1$ if and only if $A(x)>0$, and $a_1=1$ if and only if $A(x)\geq 2$. The machine accepts if $a_0=1$ and $a_1=0$ and it rejects otherwise.

This construction accepts exactly when $A(x)=1$. Since $x\in L$ if and only if $A(x)=1$, the oracle machine decides $L$. Therefore, $L\in{\bf P}^{{\bf B_{|1|\oplus}P}}$. Since $L\in{\bf US}$ was for any $L$, then we can conclude that ${\bf US}\subseteq{\bf P}^{{\bf B_{|1|\oplus}P}}$.$\qed$
\end{proof}

\begin{theorem}
${\bf US}\subseteq {\bf P}^{{\bf B_{|0|\oplus}P}}$
\end{theorem}

\begin{proof}
Let $L\in{\bf US}$. Then there is a non-deterministic polynomial time Turing machine $M$ such that, for every input $x$, $x\in L$ if and only if $M(x)$ has exactly one accepting path. Let $A(x)=\#\operatorname{acc}_M(x)$. Then $A(x)$ is a ${\bf \#P}$ function, and $x\in L$ if and only if $A(x)=1$. The machine $M$ uses exactly $p(n)$ non-deterministic bits on inputs of length $n$. We will use $p=p(|x|)$ for readability. Thus, $0\leq A(x)\leq 2^p$. Set $r=2p+1$ and $q=4p+2$. Then $r$ is odd and $q$ is even. We also have $A(x)<2^r$, $A(x)(A(x)-1)<2^r$, $(2^r-1)A(x)<2^q$, and $(2^r-1)A(x)(A(x)-1)<2^q$.

We now define $H_0(x)=4(2^q+(2^r-1)A(x))+1$. This is a ${\bf \#P}$ function since a non-deterministic machine can guess an accepting path of $M(x)$, guess a nonzero r-bit string and also the fact that ${\bf \#P}$ is closed under addition and multiplication of constants and length dependent terms. If $A(x)=0$, then $H_0(x)=4\cdot 2^q+1=2^{q+2}+1$. Its binary representation has a $1$ in bit position $q+2$, a $1$ in bit position $0$, and exactly $q+1$ zeros between them. Since $q$ is even, $q+1$ is odd. As a result, $B_0(H_0(x))\not\equiv 0\pmod 2$. If $A(x)>0$, then $1\leq A(x)<2^r$. The binary representation of $(2^r-1)A(x)$ has exactly $r$ positions occupied by $1$. Recall that $(2^r-1)A(x)=(A(x)-1)2^r+(2^r-A(x))$. Moreover, $2^r-A(x)=((2^r-1)-(A(x)-1))$. So the $r$-bit representations of $A(x)-1$ and $2^r-A(x)$ are bitwise complements. Now since $(2^r-1)A(x)<2^q$, the lower $q$ positions of $2^q+(2^r-1)A(x)$ contain exactly $r$ positions occupied by $1$ and exactly $q-r$ positions occupied by $0$. Multiplying by $4$ shifts this representation two positions to the left, and adding $1$ puts a $1$ in the lowest bit position. Thus the binary representation of $H_0(x)$ has exactly $q-r+1$ positions occupied by $0$. Since $q$ is even and $r$ is odd, $q-r+1$ is even. As a result, $B_0(H_0(x))\equiv 0\pmod 2$. Therefore, $B_0(H_0(x))\not\equiv 0\pmod 2$ if and only if $A(x)=0$.

We next define $P(x)=A(x)(A(x)-1)$. This is a ${\bf \#P}$ function because it counts ordered pairs of distinct accepting paths of $M(x)$. We have $P(x)>0$ if and only if $A(x)\geq 2$. Define $H_1(x)=2^q+(2^r-1)P(x)$.
This is a ${\bf \#P}$ function, since $2^q$ is a length dependent term and ${\bf \#P}$ is closed under multiplication and addition. If $A(x)<2$, then $P(x)=0$, so $H_1(x)=2^q$. The binary representation of $2^q$ is a single $1$ followed by exactly $q$ zeros. Since $q$ is even, $B_0(H_1(x))\equiv 0\pmod 2$. If $A(x)\geq 2$, then $P(x)>0$ and $1\leq P(x)<2^r$. The binary representation of $(2^r-1)P(x)$ has exactly $r$ positions occupied by $1$. Since $(2^r-1)P(x)<2^q$, the binary representation of $H_1(x)=2^q+(2^r-1)P(x)$ has a $1$ in bit position $q$, and its lower $q$ positions contain exactly $r$ positions occupied by $1$ and exactly $q-r$ positions occupied by $0$. Since $q$ is even and $r$ is odd, $q-r$ is odd. As a result, $B_0(H_1(x))\not\equiv 0\pmod 2$. Therefore, $B_0(H_1(x))\not\equiv 0\pmod 2$ if and only if $A(x)\geq 2$.

By tagging the query with the value of $j\in\{0,1\}$, the two functions $H_0$ and $H_1$ can be combined into a single oracle language $O\in{\bf B_{|0|\oplus}P}$. The oracle says yes on the tagged query $(x,0)$ exactly when $B_0(H_0(x))\not\equiv 0\pmod 2$, and it says yes on the tagged query $(x,1)$ exactly when $B_0(H_1(x))\not\equiv 0\pmod 2$.

The deterministic polynomial time oracle machine queries $O$ on $(x,0)$ and $(x,1)$. Let the two oracle answer bits be $a_0$ and $a_1$. By the construction above, $a_0=1$ if and only if $A(x)=0$, and $a_1=1$ if and only if $A(x)\geq 2$. The machine rejects if $a_0=1$ or $a_1=1$ and accepts otherwise. 

This construction accepts exactly when $A(x)=1$. Since $x\in L$ if and only if $A(x)=1$, the oracle machine decides $L$. Therefore, $L\in{\bf P}^{{\bf B_{|0|\oplus}P}}$. Since $L\in{\bf US}$ was for any $L$, then we can conclude that ${\bf US}\subseteq{\bf P}^{{\bf B_{|0|\oplus}P}}$.$\qed$
\end{proof}

\subsection{Containment of ${\bf \oplus P}$}

We first establish the sequences that are the class defining characteristic functions of ${\bf B_{|1|\oplus}P}$ and ${\bf B_{|0|\oplus}P}$. After that, we prove that any block of four consecutive values in either sequence can be used to determine whether its starting index is even or odd. We then use this fact to prove that ${\bf \oplus P}\subseteq {\bf P}^{{\bf B_{|1|\oplus}P}}$ and ${\bf \oplus P}\subseteq {\bf P}^{{\bf B_{|0|\oplus}P}}$.

\begin{definition}\normalfont
The Prouhet-Thue-Morse sequence \cite{M21} \cite{S96}  is the infinite binary sequence $(t(m))_{m\ge 0}$ defined by $t(m)\equiv B_1(m)\pmod 2$, where $B_1(m)$ denotes the number of $1$'s bits in the standard binary representation of $m$. Equivalently, $t(m)=0$ if $B_1(m)$ is even, and $t(m)=1$ if $B_1(m)$ is odd. The sequence begins with $t(m)_{m\ge 0}=0,$$1,1,0,1,0,0,1,1,0,0,1,0,1,1,0,\ldots$. 

It also satisfies the recursive identities $t(0)=0$, $t(2m)=t(m)$, and $t(2m+1)=1-t(m)$. In essence, the sequence contains no three consecutive equal bits. Neither $000$ nor $111$ occurs as a contiguous block.
\end{definition}

\begin{definition}\normalfont
The $B_0$-parity sequence \cite{B01} is the infinite binary sequence $(s(m))_{m\ge 0}$ defined by $s(m)\equiv B_0(m)\pmod 2$, where $B_0(m)$ denotes the number of $0$'s bits in the standard binary representation of $m$, with the convention that the standard binary representation of $0$ is $0$. Equivalently, $s(m)=0$ if $B_0(m)$ is even, and $s(m)=1$ if $B_0(m)$ is odd. The sequence begins with $s(m)_{m\ge 0}=1,$$0,1,0,0,1,1,0,1,0,0,1,0,1,1,0,\ldots$.

It also satisfies the recursive identities $s(0)=1$, $s(1)=0$, $s(2m)=1-s(m)$ for $m\ge 1$, and $s(2m+1)=s(m)$ for $m\ge 1$. In essence, the sequence contains no three consecutive equal bits. Neither $000$ nor $111$ occurs as a contiguous block.
\end{definition}

We next provide a formal proof that we can retrieve the parity of any integer $N$ by using the values for $t(N)$,$t(N+1)$,$t(N+2)$ and $t(N+3)$ or the values for $s(N)$,$s(N+1)$,$s(N+2)$ and $s(N+3)$.

\begin{theorem}
Let $t(m)=B_1(m)\pmod 2$, where $B_1(m)$ denotes the number of $1$'s bits in the standard binary representation of $m$. Then, for every $N\geq 0$, the four values $t(N),t(N+1),t(N+2),t(N+3)$ determine whether $N$ is even or odd. More precisely, $N$ is even if and only if $t(N)\neq t(N+1)$ and $t(N+2)\neq t(N+3)$.
\end{theorem}
\begin{proof}
We have the following two identities that define the sequence: for $m\geq 0$, $t(2m)=t(m)$ and $t(2m+1)=1-t(m)$. This is because the binary representation of $2m$ is obtained from the binary representation of $m$ by appending a $0$, while the binary representation of $2m+1$ is obtained by appending a $1$. Note that appending a $0$ does not change the number of $1$'s bits, while appending a $1$ changes the parity of the number of $1$'s bits.

First, we assume that $N$ is even. Then $N=2m$ for some $m\geq 0$. Thus, $N+1=2m+1$, $N+2=2(m+1)$, and $N+3=2(m+1)+1$. Using the identities above, we get $t(N)=t(2m)=t(m)$ and $t(N+1)=t(2m+1)=1-t(m)$, so $t(N)\neq t(N+1)$. Similarly, $t(N+2)=t(2(m+1))=t(m+1)$ and $t(N+3)=t(2(m+1)+1)=1-t(m+1)$, so $t(N+2)\neq t(N+3)$. Therefore, if $N$ is even, then both adjacent pairs differ.

We next assume that $N$ is odd. Then $N=2m+1$ for some $m\geq 0$. Therefore, $N+1=2(m+1)$, $N+2=2(m+1)+1$, and $N+3=2(m+2)$. Using the identities above, we get $t(N)=1-t(m)$, $t(N+1)=t(m+1)$, $t(N+2)=1-t(m+1)$, and $t(N+3)=t(m+2)$. 

Also assume for contradiction that both adjacent pairs differ. Then $t(N)\neq t(N+1)$ and $t(N+2)\neq t(N+3)$. Since $t(N)=1-t(m)$ and $t(N+1)=t(m+1)$, the inequality $t(N)\neq t(N+1)$ implies $1-t(m)\neq t(m+1)$, and $t(m)=t(m+1)$. Similarly, since $t(N+2)=1-t(m+1)$ and $t(N+3)=t(m+2)$, the inequality $t(N+2)\neq t(N+3)$ implies $1-t(m+1)\neq t(m+2)$, and $t(m+1)=t(m+2)$. Thus $t(m)=t(m+1)=t(m+2)$.

However, this is impossible. Indeed, if $m=2r$, then $t(m)=t(2r)=t(r)$ and $t(m+1)=t(2r+1)=1-t(r)$, so $t(m)\neq t(m+1)$. If $m=2r+1$, then $t(m+1)=t(2r+2)=t(r+1)$ and $t(m+2)=t(2r+3)=1-t(r+1)$, so $t(m+1)\neq t(m+2)$. Therefore, three consecutive values $t(m),t(m+1),t(m+2)$ cannot all be equal.

This contradiction shows that when $N$ is odd, it cannot be the case that both adjacent pairs differ. As a result, $t(N)\neq t(N+1)$ and $t(N+2)\neq t(N+3)$ hold simultaneously exactly when $N$ is even. Therefore, the four values $t(N),t(N+1),t(N+2),t(N+3)$ determine whether $N$ is even or odd.$\qed$
\end{proof}

\begin{theorem}
Let $s(m)=B_0(m)\pmod 2$, where $B_0(m)$ denotes the number of $0$'s bits in the standard binary representation of $m$. Then, for every $N\geq 0$, the four values $s(N),s(N+1),s(N+2),s(N+3)$ determine whether $N$ is even or odd. More precisely, $N$ is even if and only if $s(N)\neq s(N+1)$ and $s(N+2)\neq s(N+3)$.
\end{theorem}
\begin{proof}
We have the following two identities that define the sequence: for every $m\geq 1$, $s(2m)=1-s(m)$ and $s(2m+1)=s(m)$. This is because the binary representation of $2m$ is obtained from the binary representation of $m$ by appending a $0$, while the binary representation of $2m+1$ is obtained by appending a $1$. Note that appending a $0$ changes the parity of the number of $0$'s bits, while appending a $1$ does not change the number of $0$'s bits.

First, we assume that $N$ is even. If $N=0$, then the four integers are $0,1,2,3$, whose binary representations are $0,1,10,11$. Thus $s(0)=1$, $s(1)=0$, $s(2)=1$, and $s(3)=0$, so $s(N)\neq s(N+1)$ and $s(N+2)\neq s(N+3)$. Now suppose $N>0$. Then $N=2m$ for some $m\geq 1$. Therefore, $N+1=2m+1$, $N+2=2(m+1)$, and $N+3=2(m+1)+1$. Using the identities above, we get $s(N)=s(2m)=1-s(m)$ and $s(N+1)=s(2m+1)=s(m)$, so $s(N)\neq s(N+1)$. Similarly, $s(N+2)=s(2(m+1))=1-s(m+1)$ and $s(N+3)=s(2(m+1)+1)=s(m+1)$, so $s(N+2)\neq s(N+3)$. Therefore, if $N$ is even, then both adjacent pairs differ.

We next assume that $N$ is odd. If $N=1$, then the four integers are $1,2,3,4$, whose binary representations are $1,10,11,100$. Thus $s(1)=0$, $s(2)=1$, $s(3)=0$, and $s(4)=0$, so $s(N)\neq s(N+1)$, but $s(N+2)=s(N+3)$. So it is not the case that both adjacent pairs differ. Next assume $N>1$. Then $N=2m+1$ for some $m\geq 1$. Therefore, $N+1=2(m+1)$, $N+2=2(m+1)+1$, and $N+3=2(m+2)$. Using the identities above, we get $s(N)=s(m)$, $s(N+1)=1-s(m+1)$, $s(N+2)=s(m+1)$, and $s(N+3)=1-s(m+2)$.

Also assume for contradiction that both adjacent pairs differ. Then $s(N)\neq s(N+1)$ and $s(N+2)\neq s(N+3)$. Since $s(N)=s(m)$ and $s(N+1)=1-s(m+1)$, the inequality $s(N)\neq s(N+1)$ implies $s(m)\neq 1-s(m+1)$, and $s(m)=s(m+1)$. Similarly, since $s(N+2)=s(m+1)$ and $s(N+3)=1-s(m+2)$, the inequality $s(N+2)\neq s(N+3)$ implies $s(m+1)\neq 1-s(m+2)$, and $s(m+1)=s(m+2)$. Thus $s(m)=s(m+1)=s(m+2)$.

However, this is impossible for $m\geq 1$. Indeed, if $m=2r$, then $r\geq 1$, and $s(m)=s(2r)=1-s(r)$ while $s(m+1)=s(2r+1)=s(r)$, so $s(m)\neq s(m+1)$. If $m=2r+1$, then $s(m+1)=s(2(r+1))=1-s(r+1)$ and $s(m+2)=s(2(r+1)+1)=s(r+1)$, so $s(m+1)\neq s(m+2)$. Therefore, three consecutive values $s(m),s(m+1),s(m+2)$ cannot all be equal.

This contradiction shows that when $N$ is odd, it is not the case that both adjacent pairs differ. As a result, $s(N)\neq s(N+1)$ and $s(N+2)\neq s(N+3)$ hold simultaneously exactly when $N$ is even. Therefore, the four values $s(N),s(N+1),s(N+2),s(N+3)$ determine whether $N$ is even or odd.$\qed$
\end{proof}

\begin{theorem}
${\bf \oplus P}\subseteq {\bf P}^{{\bf B_{|1|\oplus}P}}$.
\end{theorem}

\begin{proof}
Let $L\in{\bf \oplus P}$. Then there exists a ${\bf \#P}$ function $F(x)$ such that $x\in L$ if and only if $F(x)\not\equiv 0\pmod 2$. We will decide whether $F(x)$ is odd using a deterministic polynomial time oracle machine with access to a language in ${\bf B_{|1|\oplus}P}$ along with the result from Theorem 6.

For each $j\in\{0,1,2,3\}$, define $F_j(x)=F(x)+j$. Each function $F_j(x)$ is a ${\bf \#P}$ function, because fixed nonnegative constants are ${\bf \#P}$ functions and ${\bf \#P}$ is closed under addition. By tagging the query with the value of $j$, all four functions $F_0,F_1,F_2,F_3$ can be combined into a single oracle language $O\in{\bf B_{|1|\oplus}P}$. The oracle says yes on the tagged query $(x,j)$ exactly when $B_1(F_j(x))\not\equiv 0\pmod 2$.

Let $t(n)=B_1(n)\pmod 2$. Then the oracle answer on $(x,j)$ is precisely $t(F(x)+j)$. The deterministic polynomial time oracle machine queries the oracle for $F_0(x)=F(x)$, $F_1(x)=F(x)+1$, $F_2(x)=F(x)+2$, and $F_3(x)=F(x)+3$. Let the four oracle-answer bits be $a_0,a_1,a_2,a_3$, where $a_j=t(F(x)+j)$. The machine rejects if $a_0\neq a_1$ and $a_2\neq a_3$. Otherwise, it accepts.

By the result from Theorem 6 applied with $N=F(x)$, the condition $a_0\neq a_1$ and $a_2\neq a_3$ holds if and only if $F(x)$ is even. Therefore, the oracle machine rejects exactly when $F(x)$ is even, and accepts exactly when $F(x)$ is odd. Since $x\in L$ if and only if $F(x)$ is odd, the oracle machine decides $L$.

Therefore, $L\in{\bf P}^{{\bf B_{|1|\oplus}P}}$. Since $L\in{\bf \oplus P}$ was for any $L$, then we can conclude that ${\bf \oplus P}\subseteq{\bf P}^{{\bf B_{|1|\oplus}P}}$.$\qed$
\end{proof}

\begin{theorem}
${\bf \oplus P}\subseteq {\bf P}^{{\bf B_{|0|\oplus}P}}$.
\end{theorem}

\begin{proof}
Let $L\in{\bf \oplus P}$. Then there exists a ${\bf \#P}$ function $F(x)$ such that $x\in L$ if and only if $F(x)\not\equiv 0\pmod 2$. We will decide whether $F(x)$ is odd using a deterministic polynomial time oracle machine with access to a language in ${\bf B_{|0|\oplus}P}$ along with the result from Theorem 7.

For each $j\in\{0,1,2,3\}$, define $F_j(x)=F(x)+j$. Each function $F_j(x)$ is a ${\bf \#P}$ function, because fixed nonnegative constants are ${\bf \#P}$ functions and ${\bf \#P}$ is closed under addition. By tagging the query with the value of $j$, all four functions $F_0,F_1,F_2,F_3$ can be combined into a single oracle language $O\in{\bf B_{|0|\oplus}P}$. The oracle says yes on the tagged query $(x,j)$ exactly when $B_0(F_j(x))\not\equiv 0\pmod 2$.

Let $s(n)=B_0(n)\pmod 2$. Then the oracle answer on $(x,j)$ is precisely $s(F(x)+j)$. The deterministic polynomial time oracle machine queries the oracle for $F_0(x)=F(x)$, $F_1(x)=F(x)+1$, $F_2(x)=F(x)+2$, and $F_3(x)=F(x)+3$. Let the four oracle-answer bits be $a_0,a_1,a_2,a_3$, where $a_j=s(F(x)+j)$. The machine rejects if $a_0\neq a_1$ and $a_2\neq a_3$. Otherwise, it accepts.

By the result from Theorem 7 applied with $N=F(x)$, the condition $a_0\neq a_1$ and $a_2\neq a_3$ holds if and only if $F(x)$ is even. Therefore, the oracle machine rejects exactly when $F(x)$ is even, and accepts exactly when $F(x)$ is odd. Since $x\in L$ if and only if $F(x)$ is odd, the oracle machine decides $L$.

Therefore, $L\in{\bf P}^{{\bf B_{|0|\oplus}P}}$. Since $L\in{\bf \oplus P}$ was for any $L$, then we can conclude that ${\bf \oplus P}\subseteq{\bf P}^{{\bf B_{|0|\oplus}P}}$.$\qed$
\end{proof}

\section{Parity based bit-counting hierarchies}

Throughout this section, we use the relativized forms of the containments and oracle simulations established previously. In particular, the constructions proving  ${\bf NP}\subseteq{\bf B_{|0|\oplus}P}$, ${\bf NP}\subseteq{\bf B_{|1|\oplus}P}$, ${\bf CoNP}\subseteq{\bf B_{|0|\oplus}P}$, ${\bf CoNP}\subseteq{\bf B_{|1|\oplus}P}$, ${\bf B_{|1|\oplus}P}\subseteq{\bf B_{|0|\oplus}P}$  and ${\bf P}^{{\bf B_{|0|\oplus}P}}={\bf P}^{{\bf B_{|1|\oplus}P}}$ continue to hold when all of the participating machines are given access to the same fixed oracle. We also use the standard composition and absorption properties of deterministic polynomial time oracle computations. Whenever such a deterministic oracle simulation is inserted inside a nondeterministic counting computation, it introduces no additional nondeterministic choices and therefore preserves the number of accepting paths of that computation. 

\subsection{Hierarchies inspired by the polynomial hierarchy}
The following four hierarchies are modeled after the polynomial hierarchy, where the $\Sigma$ versions are raw oracle towers and the $\Delta$ versions place an outer deterministic polynomial time machine over those towers. We do not introduce separate $\Pi$ levels since both ${\bf B_{|0|\oplus}P}$ and ${\bf B_{|1|\oplus}P}$ are closed under complement and this closure relativizes. 

\begin{definition}\normalfont
The $\Sigma$-${\bf B_{|0|\oplus}P}$ hierarchy, denoted ${\bf \Sigma B_{|0|\oplus}H}$, is defined as follows. The zeroth level is ${\bf \Sigma B_{|0|\oplus}^{0}}={\bf P}$. The first level is ${\bf \Sigma B_{|0|\oplus}^{1}}={\bf B_{|0|\oplus}P}$. For every $k\geq 1$, define the next level by ${\bf \Sigma B_{|0|\oplus}^{k+1}}={\bf B_{|0|\oplus}P}^{{\bf \Sigma B_{|0|\oplus}^{k}}}$. The full hierarchy is ${\bf \Sigma B_{|0|\oplus}H}=\bigcup_{k\geq 0}{\bf \Sigma B_{|0|\oplus}^{k}}$.
\end{definition}

The subsequent few levels are as follows:

-${\bf \Sigma B_{|0|\oplus}^{2}}={\bf B_{|0|\oplus}P}^{{\bf B_{|0|\oplus}P}}$

-${\bf \Sigma B_{|0|\oplus}^{3}}={\bf B_{|0|\oplus}P}^{{\bf B_{|0|\oplus}P}^{{\bf B_{|0|\oplus}P}}}$ 

-${\bf \Sigma B_{|0|\oplus}^{4}}={\bf B_{|0|\oplus}P}^{{\bf B_{|0|\oplus}P}^{{\bf B_{|0|\oplus}P}^{{\bf B_{|0|\oplus}P}}}}$

\begin{definition}\normalfont
The $\Delta$-${\bf B_{|0|\oplus}P}$ hierarchy, denoted ${\bf \Delta B_{|0|\oplus}H}$, is defined as follows. The zeroth level is ${\bf \Delta B_{|0|\oplus}^{0}}={\bf P}$. For every $k\geq 1$, define ${\bf \Delta B_{|0|\oplus}^{k}}={\bf P}^{{\bf \Sigma B_{|0|\oplus}^{k}}}$. The full hierarchy is ${\bf \Delta B_{|0|\oplus}H}=\bigcup_{k\geq 0}{\bf \Delta B_{|0|\oplus}^{k}}$.
\end{definition}

The subsequent few levels are as follows:

-${\bf \Delta B_{|0|\oplus}^{1}}={\bf P}^{{\bf B_{|0|\oplus}P}}$

-${\bf \Delta B_{|0|\oplus}^{2}}={\bf P}^{{\bf B_{|0|\oplus}P}^{{\bf B_{|0|\oplus}P}}}$

-${\bf \Delta B_{|0|\oplus}^{3}}={\bf P}^{{\bf B_{|0|\oplus}P}^{{\bf B_{|0|\oplus}P}^{{\bf B_{|0|\oplus}P}}}}$ 

-${\bf \Delta B_{|0|\oplus}^{4}}={\bf P}^{{\bf B_{|0|\oplus}P}^{{\bf B_{|0|\oplus}P}^{{\bf B_{|0|\oplus}P}^{{\bf B_{|0|\oplus}P}}}}}$

\begin{definition}\normalfont
The $\Sigma$-${\bf B_{|1|\oplus}P}$ hierarchy, denoted ${\bf \Sigma B_{|1|\oplus}H}$, is defined as follows. The zeroth level is ${\bf \Sigma B_{|1|\oplus}^{0}}={\bf P}$. The first level is ${\bf \Sigma B_{|1|\oplus}^{1}}={\bf B_{|1|\oplus}P}$. For every $k\geq 1$, define the next level by ${\bf \Sigma B_{|1|\oplus}^{k+1}}={\bf B_{|1|\oplus}P}^{{\bf \Sigma B_{|1|\oplus}^{k}}}$. The full hierarchy is ${\bf \Sigma B_{|1|\oplus}H}=\bigcup_{k\geq 0}{\bf \Sigma B_{|1|\oplus}^{k}}$.
\end{definition}

The subsequent few levels are as follows:

-${\bf \Sigma B_{|1|\oplus}^{2}}={\bf B_{|1|\oplus}P}^{{\bf B_{|1|\oplus}P}}$

-${\bf \Sigma B_{|1|\oplus}^{3}}={\bf B_{|1|\oplus}P}^{{\bf B_{|1|\oplus}P}^{{\bf B_{|1|\oplus}P}}}$

-${\bf \Sigma B_{|1|\oplus}^{4}}={\bf B_{|1|\oplus}P}^{{\bf B_{|1|\oplus}P}^{{\bf B_{|1|\oplus}P}^{{\bf B_{|1|\oplus}P}}}}$

\begin{definition}\normalfont
The $\Delta$-${\bf B_{|1|\oplus}P}$ hierarchy, denoted ${\bf \Delta B_{|1|\oplus}H}$, is defined as follows. The zeroth level is ${\bf \Delta B_{|1|\oplus}^{0}}={\bf P}$. For every $k\geq 1$, define ${\bf \Delta B_{|1|\oplus}^{k}}={\bf P}^{{\bf \Sigma B_{|1|\oplus}^{k}}}$. The full hierarchy is ${\bf \Delta B_{|1|\oplus}H}=\bigcup_{k\geq 0}{\bf \Delta B_{|1|\oplus}^{k}}$.
\end{definition}

The subsequent few levels are as follows: 

-${\bf \Delta B_{|1|\oplus}^{1}}={\bf P}^{{\bf B_{|1|\oplus}P}}$ 

-${\bf \Delta B_{|1|\oplus}^{2}}={\bf P}^{{\bf B_{|1|\oplus}P}^{{\bf B_{|1|\oplus}P}}}$ 

-${\bf \Delta B_{|1|\oplus}^{3}}={\bf P}^{{\bf B_{|1|\oplus}P}^{{\bf B_{|1|\oplus}P}^{{\bf B_{|1|\oplus}P}}}}$ 

-${\bf \Delta B_{|1|\oplus}^{4}}={\bf P}^{{\bf B_{|1|\oplus}P}^{{\bf B_{|1|\oplus}P}^{{\bf B_{|1|\oplus}P}^{{\bf B_{|1|\oplus}P}}}}}$\newline

Although ${\bf B_{|0|\oplus}P}$ and ${\bf B_{|1|\oplus}P}$ give rise to separately defined hierarchies, the next two theorems show that the corresponding $\Sigma$ hierarchies and $\Delta$ hierarchies coincide after taking the unions over all of their levels. We first prove that ${\bf \Sigma B_{|0|\oplus}H}={\bf \Sigma B_{|1|\oplus}H}$, and then prove that ${\bf \Delta B_{|0|\oplus}H}={\bf \Delta B_{|1|\oplus}H}$.

\begin{theorem}
${\bf \Sigma B_{|0|\oplus}H}={\bf \Sigma B_{|1|\oplus}H}$
\end{theorem}
\begin{proof}
Let ${\bf C_0}={\bf B_{|0|\oplus}P}$ and ${\bf C_1}={\bf B_{|1|\oplus}P}$. We know that ${\bf C_1}\subseteq{\bf C_0}$ and ${\bf C_0}\subseteq{\bf P}^{\bf C_1}$. Both containments relativize.

First, we prove ${\bf \Sigma B_{|1|\oplus}^{k}}\subseteq{\bf \Sigma B_{|0|\oplus}^{k}}$ by induction on $k$. Note that $k=0$ and $k=1$ are trivial. We now assume that ${\bf \Sigma B_{|1|\oplus}^{k}}\subseteq{\bf \Sigma B_{|0|\oplus}^{k}}$. We then obtain ${\bf \Sigma B_{|1|\oplus}^{k+1}}={\bf C_1}^{{\bf \Sigma B_{|1|\oplus}^{k}}}\subseteq{\bf C_0}^{{\bf \Sigma B_{|0|\oplus}^{k}}}={\bf \Sigma B_{|0|\oplus}^{k+1}}$. Therefore, ${\bf \Sigma B_{|1|\oplus}^{k}}\subseteq{\bf \Sigma B_{|0|\oplus}^{k}}$ for every $k\geq 0$.

Next, we prove the shifted containment ${\bf \Sigma B_{|0|\oplus}^{k}}\subseteq{\bf \Sigma B_{|1|\oplus}^{k+1}}$ for every $k\geq 1$. It is enough to prove the stronger statement ${\bf \Sigma B_{|0|\oplus}^{k}}\subseteq{\bf P}^{{\bf \Sigma B_{|1|\oplus}^{k}}}$ for every $k\geq 1$. Note that $k=1$ is trivial. We now assume that ${\bf \Sigma B_{|0|\oplus}^{k}}\subseteq{\bf P}^{{\bf \Sigma B_{|1|\oplus}^{k}}}$. We then prove ${\bf \Sigma B_{|0|\oplus}^{k+1}}\subseteq{\bf P}^{{\bf \Sigma B_{|1|\oplus}^{k+1}}}$. We have ${\bf \Sigma B_{|0|\oplus}^{k+1}}={\bf C_0}^{{\bf \Sigma B_{|0|\oplus}^{k}}}$ by definition. Since ${\bf C_0}^{\bf A}\subseteq{\bf P}^{{\bf C_1}^{\bf A}}$ for every oracle class ${\bf A}$, we get ${\bf \Sigma B_{|0|\oplus}^{k+1}}\subseteq{\bf P}^{{\bf C_1}^{{\bf \Sigma B_{|0|\oplus}^{k}}}}$. By the induction hypothesis, an oracle for ${\bf \Sigma B_{|0|\oplus}^{k}}$ can be simulated by a deterministic polynomial time machine with oracle ${\bf \Sigma B_{|1|\oplus}^{k}}$. Thus ${\bf C_1}^{{\bf \Sigma B_{|0|\oplus}^{k}}}\subseteq{\bf C_1}^{{\bf P}^{{\bf \Sigma B_{|1|\oplus}^{k}}}}$. We have ${\bf C_1}^{{\bf P}^{{\bf \Sigma B_{|1|\oplus}^{k}}}}\subseteq{\bf C_1}^{{\bf \Sigma B_{|1|\oplus}^{k}}}={\bf \Sigma B_{|1|\oplus}^{k+1}}$, since deterministic polynomial time oracle computations can be absorbed into the oracle access. Therefore, we obtain ${\bf \Sigma B_{|0|\oplus}^{k+1}}\subseteq{\bf P}^{{\bf \Sigma B_{|1|\oplus}^{k+1}}}$, which proves the stronger statement by induction. Finally, we obtain ${\bf \Sigma B_{|0|\oplus}^{k}}\subseteq{\bf P}^{{\bf \Sigma B_{|1|\oplus}^{k}}}\subseteq{\bf C_1}^{{\bf \Sigma B_{|1|\oplus}^{k}}}={\bf \Sigma B_{|1|\oplus}^{k+1}}$ for every $k\geq 1$ since ${\bf P}^{\bf A}\subseteq{\bf C_1}^{\bf A}$ for every oracle class ${\bf A}$.

Therefore, every level of the $\Sigma$-${\bf B_{|1|\oplus}P}$ hierarchy is contained in the corresponding level of the $\Sigma$-${\bf B_{|0|\oplus}P}$ hierarchy, and every positive level of the $\Sigma$-${\bf B_{|0|\oplus}P}$ hierarchy is contained in the next level of the $\Sigma$-${\bf B_{|1|\oplus}P}$ hierarchy. Taking unions over all of their levels gives ${\bf \Sigma B_{|0|\oplus}H}={\bf \Sigma B_{|1|\oplus}H}$.$\qed$
\end{proof}

\begin{theorem}
 ${\bf \Delta B_{|0|\oplus}H}={\bf \Delta B_{|1|\oplus}H}$
\end{theorem}
\begin{proof}
Let ${\bf C_0}={\bf B_{|0|\oplus}P}$ and ${\bf C_1}={\bf B_{|1|\oplus}P}$. We know that ${\bf C_0}$ and ${\bf C_1}$ are Turing equivalent, meaning ${\bf P}^{\bf C_0}={\bf P}^{\bf C_1}$. Moreover, this equivalence relativizes such that ${\bf P}^{{\bf C_0}^{\bf A}}={\bf P}^{{\bf C_1}^{\bf A}}$ for every oracle class ${\bf A}$. 

We prove by induction on $k$ that ${\bf \Delta B_{|0|\oplus}^{k}}={\bf \Delta B_{|1|\oplus}^{k}}$. Note that $k=0$ and $k=1$ are trivial. We now assume that ${\bf \Delta B_{|0|\oplus}^{k}}={\bf \Delta B_{|1|\oplus}^{k}}$. Equivalently, we have ${\bf P}^{{\bf \Sigma B_{|0|\oplus}^{k}}}={\bf P}^{{\bf \Sigma B_{|1|\oplus}^{k}}}$. We next show that ${\bf \Delta B_{|0|\oplus}^{k+1}}={\bf \Delta B_{|1|\oplus}^{k+1}}$. We have ${\bf \Delta B_{|0|\oplus}^{k+1}}={\bf P}^{{\bf \Sigma B_{|0|\oplus}^{k+1}}}={\bf P}^{{\bf C_0}^{{\bf \Sigma B_{|0|\oplus}^{k}}}}$ by definition. This equals ${\bf P}^{{\bf C_1}^{{\bf \Sigma B_{|0|\oplus}^{k}}}}$ by the relativized Turing equivalence of ${\bf C_0}$ and ${\bf C_1}$. Since ${\bf P}^{{\bf \Sigma B_{|0|\oplus}^{k}}}={\bf P}^{{\bf \Sigma B_{|1|\oplus}^{k}}}$ by the induction hypothesis, the oracle ${\bf \Sigma B_{|0|\oplus}^{k}}$ can be replaced by the Turing equivalent oracle ${\bf \Sigma B_{|1|\oplus}^{k}}$ inside an outer deterministic polynomial time computation. More specifically, each query to the oracle ${\bf \Sigma B_{|0|\oplus}^{k}}$ made within the ${\bf C_1}$ computation can be simulated deterministically using the oracle ${\bf \Sigma B_{|1|\oplus}^{k}}$; since this simulation introduces no additional non-deterministic choices, it does not change the number of accepting paths of the ${\bf C_1}$ computation. As a result, we obtain ${\bf P}^{{\bf C_1}^{{\bf \Sigma B_{|0|\oplus}^{k}}}}={\bf P}^{{\bf C_1}^{{\bf \Sigma B_{|1|\oplus}^{k}}}}={\bf P}^{{\bf \Sigma B_{|1|\oplus}^{k+1}}}={\bf \Delta B_{|1|\oplus}^{k+1}}$. Thus, ${\bf \Delta B_{|0|\oplus}^{k+1}}={\bf \Delta B_{|1|\oplus}^{k+1}}$. Then by induction ${\bf \Delta B_{|0|\oplus}^{k}}={\bf \Delta B_{|1|\oplus}^{k}}$ for every $k\geq 0$. Taking unions over all of their levels gives ${\bf \Delta B_{|0|\oplus}H}={\bf \Delta B_{|1|\oplus}H}$.$\qed$
\end{proof}

The next two theorems establish the relationship between the polynomial hierarchy and the $\Sigma$ and $\Delta$ parity based bit-counting hierarchies. We first prove that ${\bf PH}\subseteq{\bf \Sigma B_{|0|\oplus}H}$ and ${\bf PH}\subseteq{\bf \Sigma B_{|1|\oplus}H}$, and then prove that ${\bf PH}\subseteq{\bf \Delta B_{|0|\oplus}H}$ and ${\bf PH}\subseteq{\bf \Delta B_{|1|\oplus}H}$.

\begin{theorem}
${\bf PH}\subseteq{\bf \Sigma B_{|0|\oplus}H}$ and ${\bf PH}\subseteq{\bf \Sigma B_{|1|\oplus}H}$
\end{theorem}
\begin{proof}
Let ${\bf C}$ be either ${\bf B_{|0|\oplus}P}$ or ${\bf B_{|1|\oplus}P}$. We know that ${\bf NP}\subseteq{\bf C}$. This containment relativizes, so for every oracle class ${\bf A}$, we have ${\bf NP}^{\bf A}\subseteq{\bf C}^{\bf A}$. Let ${\bf \Sigma C}^{0}={\bf P}$ and, for every $k\geq 0$, let ${\bf \Sigma C}^{k+1}={\bf C}^{{\bf \Sigma C}^{k}}$. 

We prove by induction that ${\bf \Sigma_k^P}\subseteq{\bf \Sigma C}^{k}$ for every $k\geq 0$. For $k=0$, we have ${\bf \Sigma_0^P}={\bf P}={\bf \Sigma C}^{0}$. We now assume that ${\bf \Sigma_k^P}\subseteq{\bf \Sigma C}^{k}$. Since ${\bf \Sigma_{k+1}^P}={\bf NP}^{{\bf \Sigma_k^P}}$, the induction hypothesis gives ${\bf \Sigma_{k+1}^P}\subseteq{\bf NP}^{{\bf \Sigma C}^{k}}$. We get ${\bf NP}^{{\bf \Sigma C}^{k}}\subseteq{\bf C}^{{\bf \Sigma C}^{k}}={\bf \Sigma C}^{k+1}$  since ${\bf NP}^{\bf A}\subseteq{\bf C}^{\bf A}$ for every oracle class ${\bf A}$. Thus, ${\bf \Sigma_{k+1}^P}\subseteq{\bf \Sigma C}^{k+1}$. Then by induction, ${\bf \Sigma_k^P}\subseteq{\bf \Sigma C}^{k}$ for every $k\geq 0$.

Since ${\bf PH}=\bigcup_{k\geq 0}{\bf \Sigma_k^P}$, every language in ${\bf PH}$ appears at some finite level ${\bf \Sigma_k^P}$ and appears at the corresponding finite level ${\bf \Sigma C}^{k}$. Therefore, ${\bf PH}\subseteq{\bf \Sigma C H}$. Taking ${\bf C}={\bf B_{|0|\oplus}P}$ gives ${\bf PH}\subseteq{\bf \Sigma B_{|0|\oplus}H}$, and taking ${\bf C}={\bf B_{|1|\oplus}P}$ gives ${\bf PH}\subseteq{\bf \Sigma B_{|1|\oplus}H}$.$\qed$
\end{proof}

\begin{theorem}
${\bf PH}\subseteq{\bf \Delta B_{|0|\oplus}H}$ and ${\bf PH}\subseteq{\bf \Delta B_{|1|\oplus}H}$
\end{theorem}
\begin{proof}
Let ${\bf C}$ be either ${\bf B_{|0|\oplus}P}$ or ${\bf B_{|1|\oplus}P}$. We know that ${\bf NP}\subseteq{\bf C}$. This containment relativizes, so for every oracle class ${\bf A}$, we have ${\bf NP}^{\bf A}\subseteq{\bf C}^{\bf A}$.

We prove by induction that ${\bf \Sigma_k^P}\subseteq{\bf \Delta C}^{k}$ for every $k\geq 0$, where ${\bf \Delta C}^{k}$ denotes the $k$th $\Delta$-level generated from ${\bf C}$. For $k=0$, we have ${\bf \Sigma_0^P}={\bf P}={\bf \Delta C}^{0}$. We now assume that ${\bf \Sigma_k^P}\subseteq{\bf \Delta C}^{k}$. We have ${\bf \Delta C}^{k}={\bf P}^{{\bf \Sigma C}^{k}}$ by definition, where ${\bf \Sigma C}^{k}$ is the $k$th $\Sigma$-level generated from ${\bf C}$. Thus, ${\bf \Sigma_k^P}\subseteq{\bf P}^{{\bf \Sigma C}^{k}}$. We now consider ${\bf \Sigma_{k+1}^P}$. Since ${\bf \Sigma_{k+1}^P}={\bf NP}^{{\bf \Sigma_k^P}}$, the induction hypothesis gives ${\bf \Sigma_{k+1}^P}\subseteq{\bf NP}^{{\bf P}^{{\bf \Sigma C}^{k}}}$. A non-deterministic polynomial time machine with a ${\bf P}^{{\bf \Sigma C}^{k}}$ oracle can simulate that oracle by making polynomially many queries to ${\bf \Sigma C}^{k}$. Therefore, we have ${\bf NP}^{{\bf P}^{{\bf \Sigma C}^{k}}}\subseteq{\bf NP}^{{\bf \Sigma C}^{k}}$. We then get ${\bf NP}^{{\bf \Sigma C}^{k}}\subseteq{\bf C}^{{\bf \Sigma C}^{k}}={\bf \Sigma C}^{k+1}$ since ${\bf NP}^{\bf A}\subseteq{\bf C}^{\bf A}$ for every oracle class ${\bf A}$. Finally, we obtain ${\bf \Sigma C}^{k+1}\subseteq{\bf P}^{{\bf \Sigma C}^{k+1}}={\bf \Delta C}^{k+1}$. Thus ${\bf \Sigma_{k+1}^P}\subseteq{\bf \Delta C}^{k+1}$. Then by induction, ${\bf \Sigma_k^P}\subseteq{\bf \Delta C}^{k}$ for every $k\geq 0$. 

Since ${\bf PH}=\bigcup_{k\geq 0}{\bf \Sigma_k^P}$, every language in ${\bf PH}$ appears at some finite level ${\bf \Sigma_k^P}$, and appears at the corresponding finite level ${\bf \Delta C}^{k}$. Therefore, ${\bf PH}\subseteq{\bf \Delta C H}$. Taking ${\bf C}={\bf B_{|0|\oplus}P}$ gives ${\bf PH}\subseteq{\bf \Delta B_{|0|\oplus}H}$, and taking ${\bf C}={\bf B_{|1|\oplus}P}$ gives ${\bf PH}\subseteq{\bf \Delta B_{|1|\oplus}H}$.$\qed$
\end{proof}

\subsection{Alternating parity based bit-counting hierarchies}

We revise the parity based bit-counting hierarchies introduced in the previous subsection by alternating ${\bf B_{|0|\oplus}P}$ and ${\bf B_{|1|\oplus}P}$ through the oracle tower, so that the innermost class at each even level is the opposite parity based bit-counting complexity class from the one defining the hierarchy.

\begin{definition}\normalfont
The alternating-$\Sigma$-${\bf B_{|0|\oplus}P}$ hierarchy, denoted ${\bf Alt\Sigma B_{|0|\oplus}H}$, is defined as follows. The zeroth level is ${\bf Alt\Sigma B_{|0|\oplus}^{0}}={\bf P}$. For every $k\geq 1$, the $k$th level ${\bf Alt\Sigma B_{|0|\oplus}^{k}}$ is the right-associated alternating oracle tower of length $k$ whose outermost class is ${\bf B_{|0|\oplus}P}$ and whose oracle classes alternate between ${\bf B_{|1|\oplus}P}$ and ${\bf B_{|0|\oplus}P}$, so that the innermost class is ${\bf B_{|0|\oplus}P}$ when $k$ is odd and ${\bf B_{|1|\oplus}P}$ when $k$ is even. The full hierarchy is ${\bf Alt\Sigma B_{|0|\oplus}H}=\bigcup_{k\geq 0}{\bf Alt\Sigma B_{|0|\oplus}^{k}}$.

The subsequent few levels are:

${\bf Alt\Sigma B_{|0|\oplus}^{2}}={\bf B_{|0|\oplus}P}^{{\bf B_{|1|\oplus}P}}$

${\bf Alt\Sigma B_{|0|\oplus}^{3}}={\bf B_{|0|\oplus}P}^{{\bf B_{|1|\oplus}P}^{{\bf B_{|0|\oplus}P}}}$

${\bf Alt\Sigma B_{|0|\oplus}^{4}}={\bf B_{|0|\oplus}P}^{{\bf B_{|1|\oplus}P}^{{\bf B_{|0|\oplus}P}^{{\bf B_{|1|\oplus}P}}}}$

\end{definition}

\begin{definition}\normalfont
The alternating-$\Sigma$-${\bf B_{|1|\oplus}P}$ hierarchy, denoted ${\bf Alt\Sigma B_{|1|\oplus}H}$, is defined as follows. The zeroth level is ${\bf Alt\Sigma B_{|1|\oplus}^{0}}={\bf P}$. For every $k\geq 1$, the $k$th level ${\bf Alt\Sigma B_{|1|\oplus}^{k}}$ is the right-associated alternating oracle tower of length $k$ whose outermost class is ${\bf B_{|1|\oplus}P}$ and whose oracle classes alternate between ${\bf B_{|0|\oplus}P}$ and ${\bf B_{|1|\oplus}P}$, so that the innermost class is ${\bf B_{|1|\oplus}P}$ when $k$ is odd and ${\bf B_{|0|\oplus}P}$ when $k$ is even. The full hierarchy is ${\bf Alt\Sigma B_{|1|\oplus}H}=\bigcup_{k\geq 0}{\bf Alt\Sigma B_{|1|\oplus}^{k}}$.

The subsequent few levels are:

${\bf Alt\Sigma B_{|1|\oplus}^{2}}={\bf B_{|1|\oplus}P}^{{\bf B_{|0|\oplus}P}}$

${\bf Alt\Sigma B_{|1|\oplus}^{3}}={\bf B_{|1|\oplus}P}^{{\bf B_{|0|\oplus}P}^{{\bf B_{|1|\oplus}P}}}$

${\bf Alt\Sigma B_{|1|\oplus}^{4}}={\bf B_{|1|\oplus}P}^{{\bf B_{|0|\oplus}P}^{{\bf B_{|1|\oplus}P}^{{\bf B_{|0|\oplus}P}}}}$

\end{definition}

\begin{definition}\normalfont
The alternating-$\Delta$-${\bf B_{|0|\oplus}P}$ hierarchy, denoted ${\bf Alt\Delta B_{|0|\oplus}H}$, is defined as follows. The zeroth level is ${\bf Alt\Delta B_{|0|\oplus}^{0}}={\bf P}$. For every $k\geq 1$, define ${\bf Alt\Delta B_{|0|\oplus}^{k}}={\bf P}^{{\bf Alt\Sigma B_{|0|\oplus}^{k}}}$. The full hierarchy is ${\bf Alt\Delta B_{|0|\oplus}H}=\bigcup_{k\geq 0}{\bf Alt\Delta B_{|0|\oplus}^{k}}$.

The subsequent few levels are:

-${\bf Alt\Delta B_{|0|\oplus}^{1}}={\bf P}^{{\bf B_{|0|\oplus}P}}$ 

-${\bf Alt\Delta B_{|0|\oplus}^{2}}={\bf P}^{{\bf B_{|0|\oplus}P}^{{\bf B_{|1|\oplus}P}}}$ 

-${\bf Alt\Delta B_{|0|\oplus}^{3}}={\bf P}^{{\bf B_{|0|\oplus}P}^{{\bf B_{|1|\oplus}P}^{{\bf B_{|0|\oplus}P}}}}$ 

-${\bf Alt\Delta B_{|0|\oplus}^{4}}={\bf P}^{{\bf B_{|0|\oplus}P}^{{\bf B_{|1|\oplus}P}^{{\bf B_{|0|\oplus}P}^{{\bf B_{|1|\oplus}P}}}}}$

\end{definition}

\begin{definition}\normalfont
The alternating-$\Delta$-${\bf B_{|1|\oplus}P}$ hierarchy, denoted ${\bf Alt\Delta B_{|1|\oplus}H}$, is defined as follows. The zeroth level is ${\bf Alt\Delta B_{|1|\oplus}^{0}}={\bf P}$. For every $k\geq 1$, define ${\bf Alt\Delta B_{|1|\oplus}^{k}}={\bf P}^{{\bf Alt\Sigma B_{|1|\oplus}^{k}}}$. The full hierarchy is ${\bf Alt\Delta B_{|1|\oplus}H}=\bigcup_{k\geq 0}{\bf Alt\Delta B_{|1|\oplus}^{k}}$.

The subsequent few levels are:

-${\bf Alt\Delta B_{|1|\oplus}^{1}}={\bf P}^{{\bf B_{|1|\oplus}P}}$

-${\bf Alt\Delta B_{|1|\oplus}^{2}}={\bf P}^{{\bf B_{|1|\oplus}P}^{{\bf B_{|0|\oplus}P}}}$ 

-${\bf Alt\Delta B_{|1|\oplus}^{3}}={\bf P}^{{\bf B_{|1|\oplus}P}^{{\bf B_{|0|\oplus}P}^{{\bf B_{|1|\oplus}P}}}}$

-${\bf Alt\Delta B_{|1|\oplus}^{4}}={\bf P}^{{\bf B_{|1|\oplus}P}^{{\bf B_{|0|\oplus}P}^{{\bf B_{|1|\oplus}P}^{{\bf B_{|0|\oplus}P}}}}}$
\end{definition}

\begin{theorem}
${\bf Alt\Sigma B_{|0|\oplus}H}={\bf Alt\Sigma B_{|1|\oplus}H}$
\end{theorem}

\begin{proof}
Let ${\bf C_0}={\bf B_{|0|\oplus}P}$ and ${\bf C_1}={\bf B_{|1|\oplus}P}$. We know that ${\bf C_0}\subseteq{\bf P}^{\bf C_1}$ and ${\bf C_1}\subseteq{\bf P}^{\bf C_0}$. Both containments relativize. We also use the fact that ${\bf P}^{\bf D}\subseteq{\bf C_0}^{\bf D}$ and ${\bf P}^{\bf D}\subseteq{\bf C_1}^{\bf D}$ for every oracle class ${\bf D}$, since ${\bf P}\subseteq{\bf C_0}$ and ${\bf P}\subseteq{\bf C_1}$ relativize.

First, we prove the shifted containment ${\bf Alt\Sigma B_{|0|\oplus}^{k}}\subseteq{\bf Alt\Sigma B_{|1|\oplus}^{k+1}}$ for every $k\geq 1$. For $k=1$, this follows from ${\bf Alt\Sigma B_{|0|\oplus}^{1}}={\bf C_0}\subseteq{\bf P}^{\bf C_1}\subseteq{\bf P}^{{\bf P}^{\bf C_0}}={\bf P}^{\bf C_0}\subseteq{\bf C_1}^{\bf C_0}={\bf Alt\Sigma B_{|1|\oplus}^{2}}$, where the second containment uses ${\bf C_1}\subseteq{\bf P}^{\bf C_0}$. Now let $k\geq 2$. By definition, ${\bf Alt\Sigma B_{|0|\oplus}^{k}}={\bf C_0}^{{\bf Alt\Sigma B_{|1|\oplus}^{k-1}}}$. Since ${\bf C_0}^{\bf D}\subseteq{\bf P}^{{\bf C_1}^{\bf D}}$ for every oracle class ${\bf D}$, we get ${\bf Alt\Sigma B_{|0|\oplus}^{k}}\subseteq{\bf P}^{{\bf C_1}^{{\bf Alt\Sigma B_{|1|\oplus}^{k-1}}}}$. Since ${\bf C_1}^{\bf D}\subseteq{\bf P}^{{\bf C_0}^{\bf D}}$ for every oracle class ${\bf D}$, we have ${\bf C_1}^{{\bf Alt\Sigma B_{|1|\oplus}^{k-1}}}\subseteq{\bf P}^{{\bf C_0}^{{\bf Alt\Sigma B_{|1|\oplus}^{k-1}}}}={\bf P}^{{\bf Alt\Sigma B_{|0|\oplus}^{k}}}$. Therefore, ${\bf Alt\Sigma B_{|0|\oplus}^{k}}\subseteq{\bf P}^{{\bf P}^{{\bf Alt\Sigma B_{|0|\oplus}^{k}}}}={\bf P}^{{\bf Alt\Sigma B_{|0|\oplus}^{k}}}$. Since ${\bf P}^{\bf D}\subseteq{\bf C_1}^{\bf D}$ for every oracle class ${\bf D}$, it follows that ${\bf Alt\Sigma B_{|0|\oplus}^{k}}\subseteq{\bf C_1}^{{\bf Alt\Sigma B_{|0|\oplus}^{k}}}={\bf Alt\Sigma B_{|1|\oplus}^{k+1}}$. Thus ${\bf Alt\Sigma B_{|0|\oplus}^{k}}\subseteq{\bf Alt\Sigma B_{|1|\oplus}^{k+1}}$ for every $k\geq 1$.

Next, we prove the containment in the other direction. For $k=1$, this follows from ${\bf Alt\Sigma B_{|1|\oplus}^{1}}={\bf C_1}\subseteq{\bf P}^{\bf C_0}\subseteq{\bf P}^{{\bf P}^{\bf C_1}}={\bf P}^{\bf C_1}\subseteq{\bf C_0}^{\bf C_1}={\bf Alt\Sigma B_{|0|\oplus}^{2}}$, where the second containment uses ${\bf C_0}\subseteq{\bf P}^{\bf C_1}$. Now let $k\geq 2$. By definition, ${\bf Alt\Sigma B_{|1|\oplus}^{k}}={\bf C_1}^{{\bf Alt\Sigma B_{|0|\oplus}^{k-1}}}$. Since ${\bf C_1}^{\bf D}\subseteq{\bf P}^{{\bf C_0}^{\bf D}}$ for every oracle class ${\bf D}$, we get ${\bf Alt\Sigma B_{|1|\oplus}^{k}}\subseteq{\bf P}^{{\bf C_0}^{{\bf Alt\Sigma B_{|0|\oplus}^{k-1}}}}$. Since ${\bf C_0}^{\bf D}\subseteq{\bf P}^{{\bf C_1}^{\bf D}}$ for every oracle class ${\bf D}$, we have ${\bf C_0}^{{\bf Alt\Sigma B_{|0|\oplus}^{k-1}}}\subseteq{\bf P}^{{\bf C_1}^{{\bf Alt\Sigma B_{|0|\oplus}^{k-1}}}}={\bf P}^{{\bf Alt\Sigma B_{|1|\oplus}^{k}}}$. Therefore, ${\bf Alt\Sigma B_{|1|\oplus}^{k}}\subseteq{\bf P}^{{\bf P}^{{\bf Alt\Sigma B_{|1|\oplus}^{k}}}}={\bf P}^{{\bf Alt\Sigma B_{|1|\oplus}^{k}}}$. Since ${\bf P}^{\bf D}\subseteq{\bf C_0}^{\bf D}$ for every oracle class ${\bf D}$, it follows that ${\bf Alt\Sigma B_{|1|\oplus}^{k}}\subseteq{\bf C_0}^{{\bf Alt\Sigma B_{|1|\oplus}^{k}}}={\bf Alt\Sigma B_{|0|\oplus}^{k+1}}$. Thus ${\bf Alt\Sigma B_{|1|\oplus}^{k}}\subseteq{\bf Alt\Sigma B_{|0|\oplus}^{k+1}}$ for every $k\geq 1$.

Thus every positive level of the alternating-$\Sigma$-${\bf B_{|0|\oplus}P}$ hierarchy is contained in the next level of the alternating-$\Sigma$-${\bf B_{|1|\oplus}P}$ hierarchy, and every positive level of the alternating-$\Sigma$-${\bf B_{|1|\oplus}P}$ hierarchy is contained in the next level of the alternating-$\Sigma$-${\bf B_{|0|\oplus}P}$ hierarchy. Taking unions over all of their levels gives ${\bf Alt\Sigma B_{|0|\oplus}H}={\bf Alt\Sigma B_{|1|\oplus}H}$.$\qed$
\end{proof}

\begin{theorem}
${\bf Alt\Delta B_{|0|\oplus}H}={\bf Alt\Delta B_{|1|\oplus}H}$
\end{theorem}

\begin{proof}
Let ${\bf C_0}={\bf B_{|0|\oplus}P}$ and ${\bf C_1}={\bf B_{|1|\oplus}P}$. We know that ${\bf P}^{\bf C_0}={\bf P}^{\bf C_1}$, which relativizes for every oracle class ${\bf A}$, ${\bf P}^{{\bf C_0}^{\bf A}}={\bf P}^{{\bf C_1}^{\bf A}}$.

We prove by induction on $k$ that ${\bf Alt\Delta B_{|0|\oplus}^{k}}={\bf Alt\Delta B_{|1|\oplus}^{k}}$. Note that $k=0$ and $k=1$ are trivial. We now assume that ${\bf Alt\Delta B_{|0|\oplus}^{k}}={\bf Alt\Delta B_{|1|\oplus}^{k}}$. Equivalently, we have ${\bf P}^{{\bf Alt\Sigma B_{|0|\oplus}^{k}}}={\bf P}^{{\bf Alt\Sigma B_{|1|\oplus}^{k}}}$. We next show that ${\bf Alt\Delta B_{|0|\oplus}^{k+1}}={\bf Alt\Delta B_{|1|\oplus}^{k+1}}$. We have ${\bf Alt\Delta B_{|0|\oplus}^{k+1}}={\bf P}^{{\bf Alt\Sigma B_{|0|\oplus}^{k+1}}}={\bf P}^{{\bf C_0}^{{\bf Alt\Sigma B_{|1|\oplus}^{k}}}}$ by definition. This equals ${\bf P}^{{\bf C_1}^{{\bf Alt\Sigma B_{|1|\oplus}^{k}}}}$ by the relativized Turing equivalence of ${\bf C_0}$ and ${\bf C_1}$. Since ${\bf P}^{{\bf Alt\Sigma B_{|0|\oplus}^{k}}}={\bf P}^{{\bf Alt\Sigma B_{|1|\oplus}^{k}}}$ by the induction hypothesis, the oracle ${\bf Alt\Sigma B_{|1|\oplus}^{k}}$ can be replaced by the Turing equivalent oracle ${\bf Alt\Sigma B_{|0|\oplus}^{k}}$ inside an outer deterministic polynomial time computation. As a result, we obtain ${\bf P}^{{\bf C_1}^{{\bf Alt\Sigma B_{|1|\oplus}^{k}}}}={\bf P}^{{\bf C_1}^{{\bf Alt\Sigma B_{|0|\oplus}^{k}}}}={\bf P}^{{\bf Alt\Sigma B_{|1|\oplus}^{k+1}}}={\bf Alt\Delta B_{|1|\oplus}^{k+1}}$. Thus, ${\bf Alt\Delta B_{|0|\oplus}^{k+1}}={\bf Alt\Delta B_{|1|\oplus}^{k+1}}$. Then by induction ${\bf Alt\Delta B_{|0|\oplus}^{k}}={\bf Alt\Delta B_{|1|\oplus}^{k}}$ for every $k\geq 0$. Taking unions over all of their levels gives ${\bf Alt\Delta B_{|0|\oplus}H}={\bf Alt\Delta B_{|1|\oplus}H}$.$\qed$
\end{proof}

\begin{theorem}
${\bf Alt\Sigma B_{|0|\oplus}H}={\bf Alt\Sigma B_{|1|\oplus}H}={\bf \Sigma B_{|0|\oplus}H}={\bf \Sigma B_{|1|\oplus}H}$
\end{theorem}

\begin{proof}
Let ${\bf C_0}={\bf B_{|0|\oplus}P}$ and ${\bf C_1}={\bf B_{|1|\oplus}P}$. We already know that ${\bf \Sigma B_{|0|\oplus}H}={\bf \Sigma B_{|1|\oplus}H}$ and ${\bf Alt\Sigma B_{|0|\oplus}H}={\bf Alt\Sigma B_{|1|\oplus}H}$. Therefore, it is enough to prove ${\bf Alt\Sigma B_{|0|\oplus}H}={\bf \Sigma B_{|0|\oplus}H}$.

We prove by induction on $k$ that ${\bf Alt\Sigma B_{|0|\oplus}^{k}}={\bf \Sigma B_{|0|\oplus}^{k}}$. Note that $k=0$ and $k=1$ are trivial. We now assume that ${\bf Alt\Sigma B_{|0|\oplus}^{k}}={\bf \Sigma B_{|0|\oplus}^{k}}$. We then prove ${\bf Alt\Sigma B_{|0|\oplus}^{k+1}}={\bf \Sigma B_{|0|\oplus}^{k+1}}$. We have ${\bf Alt\Sigma B_{|0|\oplus}^{k+1}}={\bf C_0}^{{\bf Alt\Sigma B_{|1|\oplus}^{k}}}$ and ${\bf \Sigma B_{|0|\oplus}^{k+1}}={\bf C_0}^{{\bf \Sigma B_{|0|\oplus}^{k}}}$ by definition. Therefore, it is enough to show that the oracle ${\bf Alt\Sigma B_{|1|\oplus}^{k}}$ can be replaced by the oracle ${\bf \Sigma B_{|0|\oplus}^{k}}$ inside the outer ${\bf C_0}$ computation. 

The proof of the previous theorem established that ${\bf Alt\Delta B_{|0|\oplus}^{k}}={\bf Alt\Delta B_{|1|\oplus}^{k}}$ for every $k\geq 0$. We have ${\bf Alt\Delta B_{|0|\oplus}^{k}}={\bf P}^{{\bf Alt\Sigma B_{|0|\oplus}^{k}}}$ and ${\bf Alt\Delta B_{|1|\oplus}^{k}}={\bf P}^{{\bf Alt\Sigma B_{|1|\oplus}^{k}}}$ by definition. As a result, ${\bf P}^{{\bf Alt\Sigma B_{|0|\oplus}^{k}}}={\bf P}^{{\bf Alt\Sigma B_{|1|\oplus}^{k}}}$. By the induction hypothesis, ${\bf Alt\Sigma B_{|0|\oplus}^{k}}={\bf \Sigma B_{|0|\oplus}^{k}}$. Thus, ${\bf P}^{{\bf Alt\Sigma B_{|1|\oplus}^{k}}}={\bf P}^{{\bf \Sigma B_{|0|\oplus}^{k}}}$. So the oracles ${\bf Alt\Sigma B_{|1|\oplus}^{k}}$ and ${\bf \Sigma B_{|0|\oplus}^{k}}$ are Turing equivalent.

Therefore, an oracle for ${\bf Alt\Sigma B_{|1|\oplus}^{k}}$ can be simulated inside an outer ${\bf C_0}$ computation by using an oracle for ${\bf \Sigma B_{|0|\oplus}^{k}}$. Thus, ${\bf C_0}^{{\bf Alt\Sigma B_{|1|\oplus}^{k}}}={\bf C_0}^{{\bf \Sigma B_{|0|\oplus}^{k}}}$ and vice versa. As a result, ${\bf Alt\Sigma B_{|0|\oplus}^{k+1}}={\bf \Sigma B_{|0|\oplus}^{k+1}}$. Then by induction ${\bf Alt\Sigma B_{|0|\oplus}^{k}}={\bf \Sigma B_{|0|\oplus}^{k}}$ for every $k\geq 0$. Taking unions over all of their levels gives ${\bf Alt\Sigma B_{|0|\oplus}H}={\bf \Sigma B_{|0|\oplus}H}$. Since ${\bf Alt\Sigma B_{|0|\oplus}H}={\bf Alt\Sigma B_{|1|\oplus}H}$ and ${\bf \Sigma B_{|0|\oplus}H}={\bf \Sigma B_{|1|\oplus}H}$ were already established, then we can conclude that ${\bf Alt\Sigma B_{|0|\oplus}H}={\bf Alt\Sigma B_{|1|\oplus}H}={\bf \Sigma B_{|0|\oplus}H}={\bf \Sigma B_{|1|\oplus}H}$.$\qed$
\end{proof}

\begin{theorem}
${\bf Alt\Delta B_{|0|\oplus}H}={\bf Alt\Delta B_{|1|\oplus}H}={\bf \Delta B_{|0|\oplus}H}={\bf \Delta B_{|1|\oplus}H}$
\end{theorem}

\begin{proof}
We already know that ${\bf Alt\Delta B_{|0|\oplus}H}={\bf Alt\Delta B_{|1|\oplus}H}$ and ${\bf \Delta B_{|0|\oplus}H}={\bf \Delta B_{|1|\oplus}H}$. Therefore, it is enough to prove ${\bf Alt\Delta B_{|0|\oplus}H}={\bf \Delta B_{|0|\oplus}H}$.

We prove the level by level statement ${\bf Alt\Delta B_{|0|\oplus}^{k}}={\bf \Delta B_{|0|\oplus}^{k}}$, for every $k\geq 0$. Note that $k=0$ is trivial. We have ${\bf Alt\Delta B_{|0|\oplus}^{k}}={\bf P}^{{\bf Alt\Sigma B_{|0|\oplus}^{k}}}$ and ${\bf \Delta B_{|0|\oplus}^{k}}={\bf P}^{{\bf \Sigma B_{|0|\oplus}^{k}}}$ for $k\geq 1$ by definition. From the preceding alternating-$\Sigma$ theorem, we have the level by level equality ${\bf Alt\Sigma B_{|0|\oplus}^{k}}={\bf \Sigma B_{|0|\oplus}^{k}}$ for every $k\geq 0$. Thus, ${\bf Alt\Delta B_{|0|\oplus}^{k}} ={\bf P}^{{\bf Alt\Sigma B_{|0|\oplus}^{k}}} ={\bf P}^{{\bf \Sigma B_{|0|\oplus}^{k}}}={\bf \Delta B_{|0|\oplus}^{k}}$. Therefore, ${\bf Alt\Delta B_{|0|\oplus}^{k}}={\bf \Delta B_{|0|\oplus}^{k}}$ for every $k\geq 0$. Taking unions over all of their levels gives ${\bf Alt\Delta B_{|0|\oplus}H}={\bf \Delta B_{|0|\oplus}H}$. Since ${\bf Alt\Delta B_{|0|\oplus}H}={\bf Alt\Delta B_{|1|\oplus}H}$ and ${\bf \Delta B_{|0|\oplus}H}={\bf \Delta B_{|1|\oplus}H}$ were already established, then we can conclude that ${\bf Alt\Delta B_{|0|\oplus}H}={\bf Alt\Delta B_{|1|\oplus}H}={\bf \Delta B_{|0|\oplus}H}={\bf \Delta B_{|1|\oplus}H}$.$\qed$
\end{proof}

\subsection{Parity based bit-counting hierarchies and the counting hierarchy}

We prove that the parity based bit-counting hierarchies are contained in the counting hierarchy.

\begin{theorem}
${\bf \Sigma B_{|0|\oplus}H}={\bf \Sigma B_{|1|\oplus}H}\subseteq{\bf CH}$
\end{theorem}

\begin{proof}
Let ${\bf C_0}={\bf B_{|0|\oplus}P}$. We use the known containment ${\bf P}^{\bf C_0}\subseteq{\bf P}^{\bf PP}$. This containment relativizes, so for every oracle class ${\bf A}$, we have ${\bf P}^{{\bf C_0}^{\bf A}}\subseteq{\bf P}^{{\bf PP}^{\bf A}}$.

We prove by induction on $k$ that ${\bf \Sigma B_{|0|\oplus}^{k}}\subseteq{\bf P}^{\bf C_kP}$. Note that $k=0$ is trivial. We now assume that ${\bf \Sigma B_{|0|\oplus}^{k}}\subseteq{\bf P}^{\bf C_kP}$. We prove ${\bf \Sigma B_{|0|\oplus}^{k+1}}\subseteq{\bf P}^{\bf C_{k+1}P}$. We have ${\bf \Sigma B_{|0|\oplus}^{k+1}}={\bf C_0}^{{\bf \Sigma B_{|0|\oplus}^{k}}}$ by definition. Since ${\bf C_0}^{{\bf \Sigma B_{|0|\oplus}^{k}}}\subseteq{\bf P}^{{\bf C_0}^{{\bf \Sigma B_{|0|\oplus}^{k}}}}$, and since the known containment relativizes, we get ${\bf \Sigma B_{|0|\oplus}^{k+1}}\subseteq{\bf P}^{{\bf PP}^{{\bf \Sigma B_{|0|\oplus}^{k}}}}$. By the induction hypothesis, an oracle for ${\bf \Sigma B_{|0|\oplus}^{k}}$ can be simulated by a deterministic polynomial time machine with oracle ${\bf C_kP}$. Therefore, ${\bf PP}^{{\bf \Sigma B_{|0|\oplus}^{k}}}\subseteq{\bf PP}^{{\bf P}^{\bf C_kP}}$. Since deterministic polynomial time oracle computations can be absorbed into the oracle access, we have ${\bf PP}^{{\bf P}^{\bf C_kP}}\subseteq{\bf PP}^{\bf C_kP}={\bf C_{k+1}P}$. Thus, ${\bf PP}^{{\bf \Sigma B_{|0|\oplus}^{k}}}\subseteq{\bf C_{k+1}P}$. And we get ${\bf \Sigma B_{|0|\oplus}^{k+1}}\subseteq{\bf P}^{\bf C_{k+1}P}$. We have ${\bf P}^{\bf A}\subseteq{\bf PP}^{\bf A}$ for every oracle class ${\bf A}$. Therefore, ${\bf P}^{\bf C_kP}\subseteq{\bf PP}^{\bf C_kP}={\bf C_{k+1}P}$. Since ${\bf \Sigma B_{|0|\oplus}^{k}}\subseteq{\bf P}^{\bf C_kP}$, we get ${\bf \Sigma B_{|0|\oplus}^{k}}\subseteq{\bf C_{k+1}P}$ for every $k\geq 0$. Taking unions over all $k$ levels gives ${\bf \Sigma B_{|0|\oplus}H}\subseteq{\bf CH}$. Since ${\bf \Sigma B_{|0|\oplus}H}={\bf \Sigma B_{|1|\oplus}H}$, then we can conclude that ${\bf \Sigma B_{|0|\oplus}H}={\bf \Sigma B_{|1|\oplus}H}\subseteq{\bf CH}$.$\qed$
\end{proof}

\begin{theorem}
${\bf \Delta B_{|0|\oplus}H}={\bf \Delta B_{|1|\oplus}H}\subseteq{\bf CH}$
\end{theorem}

\begin{proof}
We use the previously established containment that, for every $k\geq 0$,
${\bf \Sigma B_{|0|\oplus}^{k}}\subseteq{\bf P}^{\bf C_kP}$,
where ${\bf C_kP}$ denotes the $k$th level of the counting hierarchy.

The claim is immediate for $k=0$. Let $k\geq 1$. We have ${\bf \Delta B_{|0|\oplus}^{k}}={\bf P}^{{\bf \Sigma B_{|0|\oplus}^{k}}}$ by definition. Since ${\bf \Sigma B_{|0|\oplus}^{k}}\subseteq{\bf P}^{\bf C_kP}$, a deterministic polynomial time machine with oracle access to ${\bf \Sigma B_{|0|\oplus}^{k}}$ can simulate each such oracle query by a deterministic polynomial time computation with oracle access to ${\bf C_kP}$. Therefore, ${\bf P}^{{\bf \Sigma B_{|0|\oplus}^{k}}}\subseteq{\bf P}^{{\bf P}^{\bf C_kP}}$. Nested deterministic polynomial time oracle computations collapse, so ${\bf P}^{{\bf P}^{\bf C_kP}}={\bf P}^{\bf C_kP}$. Thus, ${\bf \Delta B_{|0|\oplus}^{k}}\subseteq{\bf P}^{\bf C_kP}$. We have ${\bf P}^{\bf A}\subseteq{\bf PP}^{\bf A}$ for every oracle class ${\bf A}$. Applying this with ${\bf A}={\bf C_kP}$ gives ${\bf P}^{\bf C_kP}\subseteq{\bf PP}^{\bf C_kP}={\bf C_{k+1}P}$. Thus, ${\bf \Delta B_{|0|\oplus}^{k}}\subseteq{\bf C_{k+1}P}$ for every $k\geq 0$. Taking unions over all $k$ levels gives ${\bf \Delta B_{|0|\oplus}H}\subseteq{\bf CH}$. Since ${\bf \Delta B_{|0|\oplus}H}={\bf \Delta B_{|1|\oplus}H}$ was already established, then we can conclude that ${\bf \Delta B_{|0|\oplus}H}={\bf \Delta B_{|1|\oplus}H}\subseteq{\bf CH}$.$\qed$
\end{proof}

\section{Conclusion}

We studied some properties of the parity based bit-counting complexity classes. We proved that both of our parity based bit-counting complexity classes ${\bf B_{|0|\oplus}P}$ and ${\bf B_{|1|\oplus}P}$ are closed under complement. We then improved the previously known containment from ${\bf B_{|1|\oplus}P}\subseteq {\bf P}^{{\bf B_{|0|\oplus}P}}$ to ${\bf B_{|1|\oplus}P}\subseteq {\bf B_{|0|\oplus}P}$. One of the main structural questions that remains open is whether ${\bf B_{|0|\oplus}P}\subseteq {\bf B_{|1|\oplus}P}$. The best result currently known in this direction is ${\bf B_{|0|\oplus}P}\subseteq {\bf P}^{{\bf B_{|1|\oplus}P}}$. One route to overcome this might be to show that ${\bf B_{|1|\oplus}P}$ is closed under truth-table reductions (post reductions), but that seems to be challenging as well. Although we were not able to prove equivalence between our parity based bit-counting complexity classes, they are Turing equivalent, that is ${\bf P}^{{\bf B_{|0|\oplus}P}}={\bf P}^{{\bf B_{|1|\oplus}P}}$, as was previously proven in \cite{P26}. 

We also proved that ${\bf US}\subseteq {\bf P}^{{\bf B_{|1|\oplus}P}}$ and ${\bf US}\subseteq {\bf P}^{{\bf B_{|0|\oplus}P}}$ by showing that deterministic polynomial time Turing machines with oracle access to either of the two parity based bit-counting complexity classes can distinguish the case in which a nondeterministic polynomial time Turing machine has exactly one accepting path from the cases in which it has either no accepting paths or more than one accepting path.

We then showed that the acceptance criteria of ${\bf B_{|0|\oplus}P}$ and ${\bf B_{|1|\oplus}P}$ define natural infinite binary sequences over the nonnegative integers, whose values are given by their respective class defining characteristic functions. In particular, the class defining characteristic function associated with ${\bf B_{|1|\oplus}P}$ produces the well known Prouhet-Thue-Morse sequence. We then proved that, for both sequences, any block of four consecutive values corresponding to $N$, $N+1$, $N+2$, and $N+3$ can be used to determine whether the starting integer $N$ is even or odd. Using these results, we proved that ${\bf \oplus P}\subseteq {\bf P}^{{\bf B_{|1|\oplus}P}}$ and ${\bf \oplus P}\subseteq {\bf P}^{{\bf B_{|0|\oplus}P}}$, demonstrating that local combinatorial properties of bit-counting sequences can reveal arithmetic information about the underlying ${\bf \# P}$ outputs. This connection suggests a potentially broader approach in which other structural properties of bit-counting sequences could be used to recover additional information about ${\bf \# P}$ functions. Specifically, it would be interesting to determine whether similar techniques could be developed to establish containments involving the complexity classes ${\bf C_=P}={\bf ES}={\bf MNS}$ and ${\bf PP}$ within the parity based bit-counting framework. 

Finally, we defined four parity based bit-counting hierarchies inspired by the polynomial hierarchy, namely ${\bf \Sigma B_{|0|\oplus}H}$, ${\bf \Sigma B_{|1|\oplus}H}$, ${\bf \Delta B_{|0|\oplus}H}$, and ${\bf \Delta B_{|1|\oplus}H}$. We proved that ${\bf \Sigma B_{|0|\oplus}H}={\bf \Sigma B_{|1|\oplus}H}$ and ${\bf \Delta B_{|0|\oplus}H}={\bf \Delta B_{|1|\oplus}H}$ after taking the unions of their respective levels. We also showed that each of these hierarchies contains the polynomial hierarchy. This follows from the fact that the containments ${\bf NP}\subseteq {\bf B_{|0|\oplus}P}$, ${\bf NP}\subseteq {\bf B_{|1|\oplus}P}$, ${\bf CoNP}\subseteq {\bf B_{|0|\oplus}P}$, and ${\bf CoNP}\subseteq {\bf B_{|1|\oplus}P}$ relativize, allowing the corresponding containments to be lifted inductively through the levels of ${\bf PH}$. We then defined four alternating parity based bit-counting hierarchies in which ${\bf B_{|0|\oplus}P}$ and ${\bf B_{|1|\oplus}P}$ alternate across successive levels, and showed that, after taking their unions over all levels, these alternating hierarchies are equivalent to their non-alternating counterparts. Finally, we proved that all of these hierarchies are contained in the counting hierarchy ${\bf CH}$. Thus, the parity based bit-counting hierarchies occupy a natural position between ${\bf PH}$ and ${\bf CH}$, in the sense that ${\bf PH}$ is contained in each of them and each is contained in ${\bf CH}$. However, these results do not establish that either of these containments is strict, so the exact position of the parity based bit-counting hierarchies relative to ${\bf PH}$ and ${\bf CH}$ remains open.

A possible direction for future work arises from a similar connection involving the comparison based bit-counting complexity classes ${\bf B_{|0|=|1|}P}$, ${\bf B_{|0|>|1|}P}$, and ${\bf B_{|0|<|1|}P}$, whose class defining conditions correspond to the sequences \cite{K99}, \cite{Z02a}, and \cite{Z02b}, respectively. This raises the possibility that, just as the combinatorial structure of the sequences associated with the parity based bit-counting classes yielded new complexity theoretic containments, structural properties of these comparison based sequences may also provide useful tools for studying the relationships among ${\bf B_{|0|=|1|}P}$, ${\bf B_{|0|>|1|}P}$, ${\bf B_{|0|<|1|}P}$, as well as other complexity classes. Furthermore, a natural extension of this line of investigation is to define hierarchies generated by the comparison based bit-counting complexity classes themselves. It would then be interesting to determine how hierarchies based on ${\bf B_{|0|=|1|}P}$, ${\bf B_{|0|>|1|}P}$, and ${\bf B_{|0|<|1|}P}$ relate to one another and where they lie relative to established hierarchies such as ${\bf PH}$ and ${\bf CH}$.

Another potential direction for future research is to define and investigate the log space variants of the bit-counting complexity classes. For example, the log space analogues of ${\bf B_{|0|\oplus}P}$ and ${\bf B_{|1|\oplus}P}$ can be denoted by ${\bf B_{|0|\oplus}L}$ and ${\bf B_{|1|\oplus}L}$, with analogous definitions for the remaining bit-counting complexity classes. It is then interesting to determine how imposing a log space bound changes the relationships among these classes, which of the containments and equivalences established in the polynomial time setting continue to hold, and where the resulting classes lie relative to other log space classes and within ${\bf P}$. Such an investigation may also help distinguish which properties of the bit-counting classes arise from their underlying bit-counting acceptance criteria and which depend essentially on the computational power available in the polynomial time setting.

Yet another plausible direction for future work is to study bit-counting variants of the ${\bf Mod_k P}$ complexity classes defined in \cite{CH89}, and to investigate how replacing their modular counting conditions with corresponding bit-counting criteria affects their relationships with the classes considered here.

Additionally, the framework developed here suggests investigating quantifier based formulations of the parity based bit-counting hierarchies, motivated by the counting-quantifier characterization of the counting hierarchy in \cite{W86} \cite{T91}. One can introduce bit-counting quantifiers whose truth conditions depend on whether the number of witnesses has an odd number of $0$’s or $1$’s in its binary representation, and then investigate whether iterating these quantifiers produces the same levels as the oracle based hierarchies ${\bf \Sigma B_{|0|\oplus}H}$ and ${\bf \Sigma B_{|1|\oplus}H}$. Such a characterization may provide a logical interpretation of the bit-counting hierarchies and could reveal structural properties, complete problems, or relationships among their levels that are less apparent from the oracle definitions alone.

\bibliographystyle{alpha}
\bibliography{Additional_Properties_Of_Parity_Based_Bit_Counting_Complexity_Classes4}

\begin{thebibliography}{BHR00}


\bibitem[BG82]{BG82}
A.~Blass and Y.~Gurevich.
\newblock On the unique satisfiability problem.
\newblock Information and control, 55:80--88, 1982.

\bibitem[BHR00]{BHR00}
B.~Borchert, L.~Hemaspaandra, and J.~Rothe.
\newblock Restrictive acceptance suffices for equivalence problems.
\newblock LMS Journal of Computation and Mathematics, 3:86--95, 2000.

\bibitem[PZ83]{PZ83}
C.H. Papadimitriou and S.~Zachos.
\newblock Two remarks on the power of counting.
\newblock Theoretical Computer Science, pages 269--275, 1983.

\bibitem[Sim75]{S75}
J.~Simon.
\newblock On Some Central Problems of Computational Complexity.
\newblock PhD thesis, Cornell University Ithaca, 1975.

\bibitem[Val79]{V79}
L.G. Valiant.
\newblock The complexity of computing the permanent.
\newblock Theoretical Computer Science, 189--201, 1979.

\bibitem[T89]{T89}
Seinosuke Toda. 
\newblock"On the computational power of PP and (+) P." 
\newblock 30th Annual Symposium on Foundations of Computer Science. IEEE, 1989.

\bibitem[PC18]{CP18}
Tayfun Pay and James~L Cox.
\newblock An overview of some semantic and syntactic complexity classes.
\newblock Electronic Colloquium on Computational Complexity, TR18-166, 2018.

\bibitem[HO02]{HO02}
Lane~A. Hemaspaandra and Mitsunori Ogihara.
\newblock The Complexity Theory Companion.
\newblock Springer, 2002.

\bibitem[S76]{S76}
L. J. Stockmeyer
\newblock The polynomial-time hierarchy.
\newblock Theoretical Computer Science, 1--22, 1976. 

\bibitem[W86]{W86}
K.W. Wagner.
\newblock The complexity of combinatorial problems with succinct input representation. 
\newblock Acta Informatica 23, 325-–356, 1986.

\bibitem[T88]{T88}
Jacobo, Toràn. 
\newblock An oracle characterization of the counting hierarchy. 
\newblock Structure in Complexity Theory Third Annual Conference. 1988.

\bibitem[P26]{P26}
Tayfun Pay.
\newblock Bit-counting complexity classes. arXiv:2606.04406, 2026. 
\newblock https://arxiv.org/abs/2606.04406

\bibitem[M21]{M21}
Marston Morse.
\newblock Recurrent geodesics on a surface of negative curvature. 
\newblock Transactions of the American Mathematical Society 22.1, 84--100, 1921.

\bibitem[S96]{S96}
Neil James Alexander Sloane.
\newblock Sequence A010060 (Thue–Morse sequence). 
\newblock https://oeis.org/A010060
\newblock The On-Line Encyclopedia of Integer Sequences. OEIS Foundation. 1996.

\bibitem[B01]{B01}
Henry Bottomley.
\newblock Sequence A059448 (The parity of the number of zero digits when n is written in binary).
\newblock https://oeis.org/A059448
\newblock The On-Line Encyclopedia of Integer Sequences. OEIS Foundation. 2001.

\bibitem[K99]{K99}
Clark Kimberling.
\newblock Sequence A031443 (Digitally balanced numbers: positive numbers that in base 2 have the same number of 0's as 1's).
\newblock https://oeis.org/A031443
\newblock The On-Line Encyclopedia of Integer Sequences. OEIS Foundation. 1999.

\bibitem[Z02a]{Z02a}
Reinhard Zumkeller.
\newblock Sequence A072603 (Numbers which in base 2 have more 0's than 1's).
\newblock https://oeis.org/A072603
\newblock The On-Line Encyclopedia of Integer Sequences. OEIS Foundation. 2002.

\bibitem[Z02b]{Z02b}
Reinhard Zumkeller.
\newblock Sequence A072600 (Numbers which in base 2 have fewer 0's than 1's).
\newblock https://oeis.org/A072600
\newblock The On-Line Encyclopedia of Integer Sequences. OEIS Foundation. 2002.

\bibitem[CH89]{CH89}
J. Y. Cai and L. A. Hemachandra.
\newblock On the power of parity polynomial time. 
\newblock Annual Symposium on Theoretical Aspects of Computer Science. 1989.

\bibitem[T91]{T91}
Jacobo Torán.
\newblock Complexity Classes Defined by Counting Quantifiers.
\newblock Journal of the ACM, 38-3, 753–-774, 1991.

\end{thebibliography}

\end{document}